\documentclass[journal]{IEEEtran} 
\usepackage{amsmath,amsfonts}
\usepackage{algorithmic}
\usepackage{algorithm}
\usepackage{array}
\usepackage{textcomp}
\usepackage{stfloats}
\usepackage{url}
\usepackage{verbatim}
\usepackage{graphicx,subfigure}
\usepackage{amssymb}
\usepackage{cite}
\usepackage{amsthm} 
\usepackage{color}
\usepackage{diagbox,multirow,makecell}
\usepackage{hyperref}
\usepackage{booktabs}
\usepackage{multirow}

\hypersetup{
    colorlinks=true,
    linkcolor=blue,
    filecolor=magenta,      
    urlcolor=cyan,
    citecolor=blue,
}

\usepackage{epstopdf}
\newtheorem{theorem}{Theorem}

\newtheorem{lemma}{Lemma}

\newtheorem{corollary}{Corollary}

\newtheorem{proposition}{Proposition}

\definecolor{mygreen}{rgb}{0, 0.63, 0.32}
\definecolor{mybrown}{rgb}{0.7, 0.3, 0.0}

\begin{document}

\title{Mutual Coupling-Aware Beamforming in Multi-User Continuous Aperture Array Systems}

\author{Junjie~Ye, \IEEEmembership{Graduate Student Member,~IEEE,} Zhaolin~Wang, \IEEEmembership{Member,~IEEE,} Yuanwei~Liu, \IEEEmembership{Fellow,~IEEE,} 
Peichang~Zhang, \IEEEmembership{Member,~IEEE,} Lei~Huang, \IEEEmembership{Senior Member,~IEEE,}
Arumugam~Nallanathan, \IEEEmembership{Fellow,~IEEE} 
\thanks{This work has been submitted to the IEEE for possible publication. Copyright may be transferred without notice, after which this version may no longer be accessible.}
\thanks{J. Ye,  P. Zhang and L. Huang are with State Key Laboratory of Radio Frequency Heterogeneous Integration, Shenzhen University, Shenzhen, China. (e-mail: 2152432003@email.szu.edu.cn; $\lbrace$pzhang, lhuang$\rbrace$@szu.edu.cn.) \\ 
\indent Z. Wang and Y. Liu are with the Department of Electrical and Computer Engineering, the University of Hong Kong, Hong Kong, China. (e-mail: $\{$zhaolin.wang,yuanwei$\}$@hku.hk).\\
\indent A. Nallanathan is with the School of Electronic Engineering and Computer Science, Queen Mary University of London, London, UK. (e-mail: a.nallanathan@qmul.ac.uk).}}


\maketitle

\begin{abstract}
A mutual coupling–aware beamforming design for continuous aperture array (CAPA)-aided multi-user systems is investigated. First, a transmit coupling kernel is characterized to explicitly capture the mutual coupling effects inherent in CAPAs, based on which a mutual coupling–aware sum-rate maximization functional optimization problem is formulated. To address this problem, a kernel approximation (KA)-based weighted minimum mean-squared error (WMMSE) algorithm is developed. The optimal beamforming condition is derived within the WMMSE framework using the calculus of variations, while KA is employed to obtain a closed-form beamforming solution via wavenumber-domain Fourier transforms and Gauss–Legendre quadrature. Furthermore, the proposed framework is extended to CAPA-to-CAPA multiple-input multiple-output (MIMO) systems. Finally, numerical results demonstrate that: 1) the proposed algorithm achieves improved performance compared to benchmark schemes; 2) the modeled coupling effects are physically rational, where the performance of spatially discrete arrays converges to that of CAPAs; and 3) CAPA-to-CAPA MIMO systems can achieve higher degrees of freedom when the transceivers are placed in close proximity.
\end{abstract}

\begin{IEEEkeywords}
Continuous aperture array, kernel approximation, multi-user beamforming, mutual coupling,  weighted minimum mean-squared error. 
\end{IEEEkeywords}

\section{Introduction}
\IEEEPARstart{I}{n} the past two decades, multiple-input multiple-output (MIMO) technology has been recognized as a cornerstone of modern wireless communications. By deploying antenna arrays with half-wavelength antenna spacing, MIMO systems exploit spatial degrees of freedom (DoFs) to enhance spectral efficiency and link reliability \cite{MIMO}. Nevertheless, the explosive growth of user connectivity and emerging data-intensive applications has imposed unprecedented demands on throughput and service quality, rendering conventional MIMO architectures increasingly inadequate under stringent performance requirements.

To further exploit spatial resources, massive MIMO (mMIMO) has been proposed as an evolutionary extension, where hundreds or even thousands of antennas are employed to fully leverage spatial multiplexing gains \cite{mMIMO_survey}. Extensive research efforts have been devoted to characterizing and improving the performance of mMIMO systems. For instance, the work \cite{mMIMO_beamforming} developed a low-complexity beamforming framework to mitigate multi-user interference by solving a series of beam-nulling problems. The work in \cite{wideband_MIMO} extended this framework to wideband scenarios through joint user grouping, subcarrier allocation, and beamforming optimization. Learning-based beamforming strategies were further investigated in \cite{Zhilong_XL_MIMO,mMIMO_DeepUnfolding} to reduce computational complexity without sacrificing the performance. Beyond communications, mMIMO has also demonstrated considerable potential in radar sensing \cite{radar_mMIMO}, integrated sensing and communication (ISAC) systems \cite{Qihao_ISAC_mMIMO}, and satellite networks \cite{Ziyu_mMIMO_LEO_downlink}. Despite these advances, performance improvements in mMIMO largely rely on signal-domain processing, without fundamentally reshaping the underlying electromagnetic (EM) propagation environment.

Recent advances in meta-materials have enabled the development of programmable EM surfaces, which are capable of directly manipulating wave propagation. A representative architecture is the reconfigurable intelligent surface (RIS), which comprises numerous passive reflecting elements that dynamically regulate the EM waves \cite{RIS_survey}. Since RIS can introduces additional controllable spatial DoFs, RIS-assisted systems have attracted extensive research interest. Early studies investigated joint transmit and passive beamforming design, demonstrating substantial communication rate improvements \cite{RIS_WSR,RIS_WSR2}. RIS has also been leveraged for physical-layer security enhancement \cite{RIS_PLS,secure_RIS}, interference mitigation \cite{RIS_Inf,RIS_Inf2}, and ISAC system design \cite{JJ_Learning,jj_fris}. In addition to RIS, stacked intelligent metasurfaces (SIMs) have emerged as another class of programmable EM surfaces, employing multi-layer programmable structures to further enhance wave-domain controllability \cite{Survey_SIM}. The multi-layer architecture enables richer spatial manipulation capabilities, leading to improved communication performance with reduced reliance on digital beamforming \cite{An_SIM,Com_SIM}. Extensions to near-field wideband communications \cite{Qingchao_SIM}, ISAC systems \cite{ISAC_SIM}, and low-altitude economy networks \cite{SIM_LAE} have also been explored.

Motivated by the ability of programmable EM surfaces to operate directly in the wave domain, recent research has revisited array architectures beyond the conventional discrete-antenna paradigm. Continuous aperture arrays (CAPAs) have emerged as a promising evolution of mMIMO, where antenna spacing shrinks to zero to achieve continuous amplitude–phase control over the entire aperture \cite{Yuanwei_CAPA}. Theoretical performance analysis have characterized the effective DoFs \cite{CAPA_performance2}, signal-to-interference-and-noise ratio (SINR) behavior \cite{CAPA_performance}, and channel capacity between two CAPAs \cite{CAPA_capacity}. Under line-of-sight conditions, wavenumber-division multiplexing has been proposed to approach capacity limits \cite{Holo_MIMO}. Beyond performance analysis, beamforming design for CAPA systems has also attracted increasing attention. In \cite{Dai_CAP}, a Fourier-based approach was developed by discretizing the underlying functional optimization problem. In contrast, \cite{Zhaolin_CAPA_BF,Zhaolin_CAPA_MIMO} directly addressed the continuous functional formulation using functional derivatives. Owing to its enhanced spatial DoFs, CAPA has further demonstrated advantages in sensing \cite{Hao_Sensing_CAPA} and ISAC systems \cite{JJ_CAPA_ISAC}.

However, most existing works rely on idealized assumptions by neglecting mutual coupling across the continuous aperture. While such assumptions simplify analysis and algorithm design, they may lead to an overestimation of the achievable performance. A few recent studies have begun to account for mutual coupling effects. For example, \cite{Chongwen_CAPA_MutualCoupling} analyzed capacity degradation in terms of DoFs and radiation efficiency, while \cite{Pizzo_CAPA_MutualCoupling} employed wavenumber-domain analysis to characterize the joint impact of spatial correlation and mutual coupling. The work in \cite{Zhaolin_CAPA_MutualCoupling} considered a single-user CAPA system with coupling effects, where beamforming schemes were proposed to maximize directional gain. Nevertheless, existing mutual coupling–aware CAPA beamforming designs have primarily focused on relatively simple scenarios, leaving the beamforming strategies for multi-user interference environments largely unexplored.

Against the above background, this paper investigates an effective mutual coupling–aware beamforming framework for CAPA-aided multi-user communication systems. The main contributions are summarized as follows:
\begin{itemize}
    \item  A mutual coupling–aware signal model for CAPA-aided multi-user systems is established. Specifically, a uni-polarized mutual coupling kernel and the associated EM power of CAPA are derived to explicitly characterize the mutual coupling effects over the continuous aperture. Based on this model, a mutual coupling–aware sum-rate maximization problem is formulated via transmit source current pattern design.
    
    \item A kernel approximation (KA)-based weighted minimum mean-squared error (WMMSE) algorithm is proposed to solve the resultant functional optimization problem. The original sum-rate maximization problem is first transformed into an equivalent WMMSE formulation. By leveraging the calculus of variations (CoV), the necessary optimality condition for the continuous beamformer is derived. Then, the KA strategy is introduced to obtain a closed-form beamforming solution. Furthermore, a practical matrix-based implementation is developed to enable efficient numerical realization.
    
    \item  The proposed mutual coupling model and beamforming algorithm are extended to the CAPA-to-CAPA MIMO scenario. In particular, a mutual coupling–aware achievable rate maximization problem is formulated for CAPA-to-CAPA MIMO systems. The symbol estimator is adapted to the continuous aperture case, based on which the corresponding mean-squared error (MSE) matrix is derived. With this formulation, the weighted MSE minimization problem and the associated algorithmic adaption are presented.

    \item Comprehensive simulation results are provided for performance evaluation. Key findings include: 1) the proposed KA-based WMMSE algorithm effectively improves sum-rate, especially for large apertures and high carrier frequencies; 2) when mutual coupling is properly accounted for, the sum-rate performance of spatially discrete arrays (SPDA) converges to the CAPA limit as antenna spacing decreases; and 3) in CAPA-to-CAPA MIMO systems, higher spatial DoFs can be exploited as the distance between the CAPA transmitter and receiver decreases.
\end{itemize}

The remainder of this paper is organized as follows. Section~\ref{sec:model} presents the system model of the CAPA-aided multi-user communication system with mutual coupling effects and formulates a sum-rate maximization problem. Section~\ref{sec:algorithm} derives the optimal beamforming solution based on the KA-based WMMSE framework. Section~\ref{sec:extension} extends the proposed algorithm to the CAPA-to-CAPA MIMO scenario. Section~\ref{simulation} provides numerical results to validate the effectiveness of the proposed approach, and Section~\ref{conclusions} concludes the paper.

\textit{Notation}: Regular, bold lowercase, and uppercase letters denote scalars, vectors, and matrices, respectively. The symbols $\mathbb{C}^{M\times N}$ and $\mathbb{R}^{M\times N}$ represent complex and real space with a dimension of $M\times N$. Moreover, $(\cdot)^{-1}$, $(\cdot)^\mathrm{H}$, $(\cdot)^\mathrm{T}$, and $(\cdot)^*$ denote the inverse, conjugate transpose, transpose, and conjugate operations, respectively. $\int_{\mathcal{S}} f(\mathbf{s}) d\mathbf{s}$ is the integral of function $f(\mathbf{s})$ over the field $\mathcal{S}$. The ceiling operator is denoted by $\lceil \cdot \rceil$, while $|\cdot|$ and $|\cdot|$ represent the absolute value and norm, respectively. Besides, $\nabla$ stands for the nabla operator. Finally, $\Re\{x\}$ and $\Im\{x\}$ denote the real and imaginary part of $x$, while $\jmath$ is the imaginary unit.

\section{System Model and Problem Formulation} \label{sec:model}
\subsection{System Model}
A CAPA-aided downlink multi-user communication system that accounts for mutual coupling effect is considered. In this system, a CAPA $\mathcal{S}_{\mathrm{T}}$ is deployed at the transmitter to serve $K$ users, where each user is equipped with a single uni-polarized antenna. The CAPA is placed on the $x$–$y$ plane and has an area of $|\mathcal{S}_{\mathrm{T}}| = L_x \cdot L_y = A_{\mathrm{T}}$, where $L_x$ and $L_y$ denote the physical lengths of the aperture along the $x$- and $y$-axes, respectively.

\subsubsection{Transmit Signal}
We denote the coordinate of an arbitrary point on the CAPA by $\mathbf{s} = \left[s_x, s_y, 0\right]^{\mathrm{T}} \in \mathcal{S}_{\mathrm{T}}$, and assume that the CAPA is uni-polarized along the $y$-axis. Under this configuration, the source current density over the CAPA aperture can be expressed as
\begin{align} \label{transmit_signal}
    \mathbf{j}_\mathrm{t}(\mathbf{s}) = j_\mathrm{t}(\mathbf{s})\mathbf{u}_\mathrm{T} \in 
\mathbb{C}^{3 \times 1},
\end{align}
where $j_{\mathrm{t}}(\mathbf{s}) \in \mathbb{C}$ denotes the uni-polarized component, and $\mathbf{u}_{\mathrm{T}} = [0,1,0]^{\mathrm{T}}$ represents the polarization direction of the transmitter. To support the simultaneous transmission to all users, $K$ independent data streams are transmitted. Accordingly, the total source current distribution $j_{\mathrm{t}}(\mathbf{s})$ can be expressed as a linear superposition of the $K$ independent data streams, i.e.,
\begin{align} \label{transmit_signal_2}
    j_\mathrm{t}(\mathbf{s}) = \sum_{k=1}^{K} w_k(\mathbf{s}) c_k = \mathbf{w} (\mathbf{s})\mathbf{c}.
\end{align}
Here, we denote $\mathbf{w}(\mathbf{s}) = [w_1(\mathbf{s}), \cdots, w_K(\mathbf{s})] \in \mathbb{C}^{1 \times K}$ and $\mathbf{c} = [c_1, \cdots, c_K] \in \mathbb{C}^{K \times 1}$, where $w_k(\mathbf{s})$ and $c_k$ represent the source current pattern and the communication symbol of the $k$-th user, respectively. Without loss of generality, the symbols in $\mathbf{c}$ are assumed to be mutually independent and normalized to unit power, i.e., $\mathbb{E}\left[\mathbf{c}\mathbf{c}^{\mathrm{H}}\right] = \mathbf{I}_K$.

\subsubsection{EM Power}
The total transmit EM power consists of two primary components, including the radiated power $P_{\mathrm{rad}}$ and the dissipated power $P_{\mathrm{diss}}$. Given the spatially continuous nature of the CAPA, the current at each point is affected by the electric fields generated by all other points, recognized as mutual coupling effects. As a result, $P_{\mathrm{rad}}$ accounts not only for the energy radiated into free space, but also for the work done against the induced electric field over the aperture. Additionally, owing to the finite conductivity of the surface, part of the supplied energy is dissipated as heat, which leads to $P_{\mathrm{diss}}$. Therefore, the total EM power can be modeled as
\begin{align}\label{em_power}
    P_{\mathrm{em}} = P_{\mathrm{rad}}+P_{\mathrm{diss}}. 
\end{align} 

For the radiated power $P_{\mathrm{rad}}$,  it is exerted by $\mathbf{j}_\mathrm{t}(\mathbf{s})$ to work against the radiated field $\mathbf{e}_\mathrm{rad}(\mathbf{s})$ across the entire aperture, which can be characterized as \cite{Pizzo_CAPA_MutualCoupling}
\begin{align} \label{rad_power}
    P_{\mathrm{rad}}=\frac{1}{2}\Re\left\{\int_{\mathcal{S}_{\mathrm{T}}}  \mathbb{E}\left[\mathbf{j}_\mathrm{t}^{\mathrm{H}}(\mathbf{s}) \mathbf{e}_\mathrm{rad}(\mathbf{s}) \right]d\mathbf{s} \right\}.
\end{align}
The radiated field $\mathbf{e}_\mathrm{rad}(\mathbf{s})$ is governed by the inhomogeneous Helmholtz wave equation, given by
\begin{align}\label{rad_field}
    \mathbf{e}_\mathrm{rad}(\mathbf{s})  &= \int_{\mathcal{S}_{\mathrm{T}}} \mathbf{G}(\mathbf{s}-\mathbf{z}) \mathbf{j}_\mathrm{t}(\mathbf{z}) d\mathbf{z}.
\end{align}
Here, $\mathbf{G}(\mathbf{s}) \in \mathbb{C}^{3\times 3}$ corresponds to the Green’s function, i.e.,
\begin{align} 
    \mathbf{G}(\mathbf{s}) = -\jmath \kappa_0 Z_0 \left(\mathbf{I}_3+\frac{1}{\kappa_0^2}\nabla^2\right) g(\mathbf{s}),
\end{align}
where $\kappa_0=2\pi/\lambda$, $\lambda$ and $Z_0$ denote the wavenumber, the signal wavelength and the free space impedance, respectively. Furthermore, $g(\mathbf{s})$ represents the scalar Green's function, which is given by
\begin{align}
    g(\mathbf{s}) = \frac{e^{\jmath \kappa_0 \|\mathbf{s}\|}}{4\pi \|\mathbf{s}\|}.
\end{align}
By substituting (\ref{rad_field}) to (\ref{rad_power}), $P_{\mathrm{rad}}$ can be rewritten as
\begin{align}\label{rad_power2}
    P_{\mathrm{rad}}=\frac{1}{2}\Re\left\{\int_{\mathcal{S}_{\mathrm{T}}}  \int_{\mathcal{S}_{\mathrm{T}}} \mathbb{E}\left[\mathbf{j}_\mathrm{t}^{\mathrm{H}}(\mathbf{s}) \mathbf{G}(\mathbf{s}-\mathbf{z}) \mathbf{j}_\mathrm{t}(\mathbf{z})\right] d \mathbf{z} d\mathbf{s} \right\}.
\end{align}

On the other hand,  the dissipated power $P_{\mathrm{diss}}$ arises since the practical transmit surfaces are not ideal conductors. These non-ideal conductors exhibit an inherent surface resistance $Z_s\in \mathbb{R}$, leading to the loss of a portion of the supplied energy as heat. Suppose that the CAPA is a good conductor, $Z_s$ can be characterized as  \cite{pozar2011microwave}
\begin{align}
    Z_s=\sqrt{\frac{\pi f_c \mu_s}{\sigma_s}},
\end{align}
in which  $f_c$, $\mu_s$ and $\sigma_s$ denote the carrier frequency, the surface permeability, and the surface conductivity, respectively. Accordingly, the  dissipated power can be characterized as \cite{pozar2011microwave}
\begin{align} \label{diss_power}
    P_{\mathrm{diss}}=\frac{Z_s}{2} \int_{\mathcal{S}_{\mathrm{T}}}  \left \| \mathbf{j}_\mathrm{t} (\mathbf{s}) \right\|^2  d\mathbf{s} .
\end{align}

By substituting (\ref{rad_power2}) and (\ref{diss_power}) into (\ref{em_power}), the EM power can be recast as 
\begin{align} \label{em_power2}
    P_{\mathrm{em}}&=\frac{1}{2}\Re\left\{\int_{\mathcal{S}_{\mathrm{T}}}  \int_{\mathcal{S}_{\mathrm{T}}} \mathbb{E}\left[\mathbf{j}_\mathrm{t}^{\mathrm{H}}(\mathbf{s}) \mathbf{C}(\mathbf{s}-\mathbf{z}) \mathbf{j}_\mathrm{t}(\mathbf{z})\right] d \mathbf{z} d\mathbf{s} \right\},
\end{align} 
where $\mathbf{C}(\mathbf{s})$ is defined as a mutual coupling kernel independent of $\mathbf{j}_\mathrm{t} (\mathbf{s})$, given by
\begin{align}
    \mathbf{C}(\mathbf{s}) = Z_s \delta(\mathbf{s}) \mathbf{I}_3 + \mathbf{G}(\mathbf{s}).
\end{align}
Here, $\delta(\mathbf{s})$ denotes the Dirac delta function. Furthermore, by substituting  (\ref{transmit_signal}) into (\ref{em_power2}), the EM power can be expressed as
\begin{align}
    P_{\mathrm{em}}
    &=  {\frac{1}{2}\int_{\mathcal{S}_{\mathrm{T}}} \int_{\mathcal{S}_{\mathrm{T}}} \mathbf{w} (\mathbf{s})  {c}_\mathrm{T}(\mathbf{s}-\mathbf{z}) \mathbf{w}^\mathrm{H} (\mathbf{z})  d \mathbf{z} d\mathbf{s}},
\end{align} 
where   ${c}_\mathrm{T}(\mathbf{s}) \in \mathbb{C}$ denotes the uni-polarized coupling kernel, expressed as 
\begin{align} \label{mutual_kernel}
    {c}_\mathrm{T}(\mathbf{s}) &= \Re\left\{\mathbf{u}_\mathrm{T}^\mathrm{T} \mathbf{C}(\mathbf{s}) \mathbf{u}_\mathrm{T} \right\} \notag \\
    &=  {Z_s \delta(\mathbf{s})}  +  {\kappa_0 Z_0 \left( \phi(\mathbf{s})+\frac{1}{\kappa_0^2} \partial _y^2\phi(\mathbf{s})\right)},
\end{align}
with
\begin{align}
    \phi(\mathbf{s}) = \Im\{g(\mathbf{s})\} = \frac{\sin(\kappa_0 \|\mathbf{s}\|)}{4\pi \|\mathbf{s}\|}.
\end{align}
In (\ref{mutual_kernel}), the first and second terms represent the dissipated and radiated mutual coupling kernels, respectively. For notational brevity, these components are denoted by
\begin{subequations}
\begin{align}
{c}_\mathrm{diss}(\mathbf{s})
    &= Z_s \delta(\mathbf{s}),\\
    {c}_\mathrm{rad}(\mathbf{s})
    &= \kappa_0 Z_0 \left( \phi(\mathbf{s})+\frac{1}{\kappa_0^2} \partial _y^2\phi(\mathbf{s})\right).
\end{align}    
\end{subequations}

\subsubsection{Received Signal}
Let $\mathbf{r}_k$ and $\mathbf{u}_{k}$ be the location and the polarization direction of user $k$, respectively. Consequently, the signal received by user $k$ can be expressed as \cite{Zhaolin_CAPA_BF}
\begin{align} \label{user_receive_signal}
    y_k &= \mathbf{u}_{k}^\mathrm{T} \mathbf{e}_{\mathrm{rad}}(\mathbf{r}_k) + z_k \notag \\
    &=\int_{\mathcal{S}_{\mathrm{T}}} \mathbf{u}_{k}^\mathrm{T}\mathbf{G}(\mathbf{r}_k-\mathbf{s})\mathbf{u}_\mathrm{T}  j_\mathrm{t}(\mathbf{s})d\mathbf{s} +z_k \notag \\
    &= \int_{\mathcal{S}_{\mathrm{T}}} h_k(\mathbf{s})   {j}_\mathrm{t}(\mathbf{s}) d\mathbf{s} + z_k. 
\end{align}
Here, $z_k \sim \mathcal{CN}(0, \sigma_k^2)$ is the additive white Gaussian noise (AWGN)  at user $k$. Furthermore,  $h_k(\mathbf{s})$ represents the channel response from CAPA to user $k$, formulated as
\begin{align} \label{channel}
    h_k( \mathbf{s}) &= \mathbf{u}_{k}^\mathrm{T} \mathbf{G}(\mathbf{r}_k-\mathbf{s}) \mathbf{u}_\mathrm{T} \notag \\
    & \overset{\mathrm{(a)}}{=} -\jmath \kappa_0 Z_0 \left(g(\mathbf{r}_k-\mathbf{s})+\frac{1}{\kappa_0^2} \partial _y^2g(\mathbf{r}_k-\mathbf{s})\right).
\end{align}
In step (a), the polarization direction of user $k$ is assumed to be aligned with the $y$-axis, i.e., $\mathbf{u}_{k}=[0,1,0]^\mathrm{T}$. Based on the signal model in (\ref{user_receive_signal}), the SINR of user $k$ is formulated as 
\begin{align}
    \gamma_k &= \frac{| {e}_{k,k}  |^2}{ \sum_{i\neq k }^{K} | {e}_{k,i} |^2+\sigma_k^2},
\end{align}
where we denote $ {e}_{k,i}  = \int_{\mathcal{S}_{\mathrm{T}}} h_k( \mathbf{s}) {w}_i(\mathbf{s}) d \mathbf{s}$.
Accordingly, the sum-rate of this system can be given by  
\begin{align}\label{sum_rate}
    R=\sum_{k=1}^{K}\log \left( 1+ \gamma_k\right).
\end{align}

\subsection{Problem Formulation}
In this paper, we aim to maximize the sum-rate of the system by optimizing the transmit current source pattern $\mathbf{w}(\mathbf{s})$. Specifically, the optimization problem is formulated as   
\begin{subequations} \label{problem_1}
    \begin{align}
        \max_{\mathbf{w}(\mathbf{s})} \quad &\sum_{k=1}^{K}\log \left( 1+ \gamma_k\right)\\
        \label{power_constraint}
        \mathrm{s.t.} \quad & \frac{1}{2}\int_{\mathcal{S}_{\mathrm{T}}} \int_{\mathcal{S}_{\mathrm{T}}} \mathbf{\mathbf{w}(\mathbf{s}) }c_{\mathrm{T}} (\mathbf{s}-\mathbf{z}) \mathbf{w}^\mathrm{H}(\mathbf{z}) d \mathbf{z} d \mathbf{s} \le P_{\mathrm{T}},
    \end{align}
\end{subequations} 
where the constraint (\ref{problem_1}{b}) specifies that the total transmit power is limited by the power budget $P_{\mathrm{T}}$. It is challenging to solve the problem (\ref{problem_1}) owing to two primary reasons. First, the objective function is non-convex with respect to (w.r.t.) the optimization variable, which is a vector of functions. Furthermore, the mutual coupling effect introduces a complex double integral within the power constraint, significantly complicating the optimization process.

\section{Proposed Algorithm}  \label{sec:algorithm}
In this section, a KA-based WMMSE algorithm is proposed to solve problem (\ref{problem_1}). Specifically, we first transform the original problem into an unconstrained MSE minimization problem. Subsequently, the optimal closed-form solution of the beamformer is derived by leveraging CoV and KA. Furthermore, a matrix-based implementation scheme is developed to facilitate the calculation of the optimal beamformer.  

\subsection{KA-based WMMSE Algorithm} \label{sec:KA_WMMSE}
To solve problem (\ref{problem_1}), we first transform it into an unconstrained optimization problem by invoking the following lemma.  
\begin{lemma} \label{equal_power_lemma}
    \normalfont
    \emph{(Equivalent Problem)} Solving problem (\ref{problem_1}) can be equivalently reformulated as finding the optimal solution to the following unconstrained optimization problem:
    \begin{align} \label{unconstrained_problem}
     \max_{{\mathbf{w}}(\mathbf{s})} \quad & \widetilde{R}=\sum_{k=1}^{K}\log \left( 1+ \widetilde{\gamma}_k\right),
    \end{align}  
    where 
        \begin{subequations}\label{expression_SINR_sigma}
            \begin{align}
        \widetilde{\gamma}_k &=   \frac{| {e}_{k,k}   |^2}{ \sum_{i\neq k }^{K} | {e}_{k,i}  |^2+\widetilde{\sigma}_k^2}, \\
        \widetilde{\sigma}_k^2 &= \frac{\sigma_k^2}{2 P_\mathrm{T}}\int_{\mathcal{S}_{\mathrm{T}}} \int_{\mathcal{S}_{\mathrm{T}}}  {\mathbf{w}}(\mathbf{s}) c_{\mathrm{T}} (\mathbf{s}-\mathbf{z}) {\mathbf{w}}^\mathrm{H}(\mathbf{z}) d \mathbf{z} d \mathbf{s}.
    \end{align}
    \end{subequations}
     Let $\widetilde{\mathbf{w}}^\star(\mathbf{s})$ denote the optimal solution to the unconstrained optimization problem \eqref{unconstrained_problem}. Accordingly, the optimal solution to problem \eqref{problem_1} can be determined by 
    \begin{align}\label{scale_w}
     \!\!\!{\mathbf{w}}^\star(\mathbf{s})\!=\!\sqrt{\frac{2 P_{\mathrm{T}}}{\int_{\mathcal{S}_{\mathrm{T}}} \!\int_{\mathcal{S}_{\mathrm{T}}}  \!\widetilde{\mathbf{w}}^\star(\mathbf{s}) c_{\mathrm{T}} (\mathbf{s}-\mathbf{z}) (\widetilde{\mathbf{w}}^\star(\mathbf{z})) ^\mathrm{H} d \mathbf{z} d \mathbf{s}}}\widetilde{\mathbf{w}}^\star(\mathbf{s}).
\end{align}
\end{lemma}

\begin{IEEEproof}
     {It is readily proved that full power should be exploited to maximize the sum-rate, i.e., the optimum is attained when the power constraint holds with equality. Clearly, the solution of \eqref{scale_w} satisfies this full-power transmission condition. Furthermore, by substituting \eqref{scale_w} into \eqref{problem_1}, it can be verified that the objective value of \eqref{problem_1} is identical to that of \eqref{unconstrained_problem}, thereby completing the proof.}
\end{IEEEproof}

For the unconstrained problem (\ref{unconstrained_problem}), an equivalent received signal model for user $k$ can be formulated as 
\begin{align} \label{equivalent_signal_model}
    \widetilde{y}_k  &=  \int_{\mathcal{S}_{\mathrm{T}}} h_k( \mathbf{s})   {j}_\mathrm{t}(\mathbf{s}) d\mathbf{s} +  \widetilde{z}_k  \overset{\mathrm{(a)}}{=} \mathbf{e}_k  \mathbf{c}+ \widetilde{ {z}}_k,
\end{align}
where $\widetilde{ {z}}_k \sim \mathcal{CN}(0, \widetilde{\sigma}_k^2) $ is the equivalent noise. More particularly, step (a) is obtained following equation (\ref{transmit_signal_2}), and $\mathbf{e}_k  = [{e}_{k,1}, \cdots,{e}_{k,K}] $ is introduced to streamline the expression. 

Given the equivalent received signal model in (\ref{equivalent_signal_model}) and the unconstrained problem in (\ref{unconstrained_problem}), the WMMSE framework is employed to transform the sum-rate maximization problem into a weighted MSE minimization problem. Specifically, we introduce a receiver ${v}_k(\mathbf{s}) \in \mathbb{C}$ to estimate the data symbol for user $k$. Accordingly, the estimated symbol can be given by 
\begin{align}
    \hat{{c}}_k&=  {v}_k^* \widetilde{y}_k = {v}_k^*\mathbf{e}_k   \mathbf{c} + \hat{{z}}_k. 
\end{align}
where $\hat{{z}}_k =   {v}_k^* \widetilde{ {z}}_k  \sim \mathcal{CN}\left( {0},\widetilde{\sigma}_k^2   |{v}_k|^2 \right)$ represents the estimation noise. The MSE between the estimated symbol and the true symbol for user $k$ is then given by 
\begin{align}\label{MSE}
    \varepsilon_k=&\mathbb{E}\left[\left| \hat{{c}}_k-c_k\right|^2\right] \notag\\
        =& |v_k|^2 \mathbf{e}_k  \mathbf{e}_k^\mathrm{H} + 1-2\Re\left\{ {v}_k^* e_{k,k}\right\} +\widetilde{\sigma}_k^2 |{v}_k|^2 \notag\\
    =& \sum_{i=1}^{K} \left|\int_{\mathcal{S}_{\mathrm{T}}} {g}_k^*(\mathbf{s})  {w}_i(\mathbf{s}) d\mathbf{s} \right|^2 + 1 \notag\\
    &~~ -2\Re\left\{\int_{\mathcal{S}_{\mathrm{T}}} {g}_k^*(\mathbf{s})  {w}_k(\mathbf{s}) d\mathbf{s}\right\} +\widetilde{\sigma}_k^2 |{v}_k|^2,
\end{align}
where we denote ${g}_k(\mathbf{s})= {v}_k h_k^*  ( \mathbf{s})$. 
 The optimal MMSE receiver $ {v}_k^{\mathrm{MMSE}} $ can be obtained by solving $\partial \varepsilon_k/\partial {v}_k = 0$,  which yields 
\begin{align} \label{MMSE_receiver}
    {v}_k^{\mathrm{MMSE}}  &= \frac{e_{k,k}}{\mathbf{e}_k  \mathbf{e}_k^\mathrm{H}+ \widetilde{\sigma}_k^2}.
\end{align}
By substituting (\ref{MMSE_receiver}) to (\ref{MSE}), the minimum MSE achieved by $ {v}_k^{\mathrm{MMSE}} $ can be given by 
\begin{align} \label{MMSE}
   \varepsilon_k^\mathrm{MMSE} &= 1- \frac{|e_{k,k}|^2}{\widetilde{\sigma}_k^2+\mathbf{e}_k \mathbf{e}_k^\mathrm{H} }. 
\end{align}
Comparing (\ref{unconstrained_problem}) and (\ref{MMSE}),  a rate-MMSE relationship of $\widetilde{R}=\sum_{k=1}^{K}\log \left( \mu_k\right)$ can be observed, where $\mu_k=(\varepsilon_k^\mathrm{MMSE})^{-1}$. According to \cite{WMMSE}, problem (\ref{unconstrained_problem}) can achieve the same optimal solution as the following weighted MSE minimization problem: 
\begin{align} \label{WMMSE_formulation_w}
     \min_{ {\mathbf{w}}(\mathbf{s})} \quad & f_c =\sum_{k=1}^{K} \mu_k \varepsilon_k.
\end{align}
\begin{figure*}[!h]
\begin{subequations}\label{MMSE_problem_2_3}
      \begin{align} 
    f_c 
      & =  \sum_{k=1}^{K}\sum_{i=1}^{K}   \mu_i  \left|\int_{\mathcal{S}_{\mathrm{T}}}\! {g}_i^*(\mathbf{s})  {w}_k(\mathbf{s}) d\mathbf{s} \right|^2 \! \!-\! \sum_{k=1}^{K} \! 2\Re\left\{ \!\int_{\mathcal{S}_{\mathrm{T}}} \!\mu_k{g}_k^*(\mathbf{s})  {w}_k(\mathbf{s}) d\mathbf{s}\right\}  
     \!+ \! \sum_{k=1}^{K} \!\mu_k \frac{\sigma_k^2}{2 P_\mathrm{T}}|{v}_k|^2 \!\int_{\mathcal{S}_{\mathrm{T}}}  \!\int_{\mathcal{S}_{\mathrm{T}}} \! \! {\mathbf{w}}(\mathbf{s}) c_{\mathrm{T}} (\mathbf{s}\!-\!\mathbf{z})  {\mathbf{w}}^\mathrm{H}(\mathbf{z}) d \mathbf{z} d \mathbf{s}  \\
     & =   \int_{\mathcal{S}_{\mathrm{T}}} \int_{\mathcal{S}_{\mathrm{T}}}    \mathbf{g}(\mathbf{s}) \mathbf{U} \mathbf{g}^\mathrm{H}(\mathbf{z}) {\mathbf{w}}(\mathbf{z})   {\mathbf{w}}^\mathrm{H}(\mathbf{s})  d\mathbf{s} d\mathbf{z}  -  2\Re\left\{\int_{\mathcal{S}_{\mathrm{T}}}  {\mathbf{w}}(\mathbf{s}) \mathbf{U} \mathbf{g}^\mathrm{H}(\mathbf{s})d\mathbf{s}\right\} + \frac{1}{\beta}\int_{\mathcal{S}_{\mathrm{T}}} \int_{\mathcal{S}_{\mathrm{T}}}   {\mathbf{w}}(\mathbf{s}) c_{\mathrm{T}} (\mathbf{s}-\mathbf{z})  {\mathbf{w}}^\mathrm{H}(\mathbf{z}) d \mathbf{z} d \mathbf{s}
\end{align}  
\end{subequations}
\hrulefill
\end{figure*}

To better examine the objective function, it is expanded in (\ref{MMSE_problem_2_3}{a}) and transformed into the matrix form in (\ref{MMSE_problem_2_3}{b}), where we denote  
\begin{subequations}\label{g_U_beta}
\begin{align}
 \mathbf{g}(\mathbf{s})&=[g_1(\mathbf{s}),\cdots,g_K(\mathbf{s})] \in \mathbb{C}^{1 \times K}, \\
    \mathbf{U}&=\mathrm{diag}(\mu_1,\cdots,\mu_K)\in \mathbb{R}^{K \times K},\\
    \frac{1}{\beta} &=\sum_{k=1}^{K} \mu_k \frac{\sigma_k^2}{2 P_\mathrm{T}}|{v}_k|^2.
\end{align}    
\end{subequations}
It can be observed from the (\ref{MMSE_problem_2_3}{b}) that the objective function is a convex quadratic functional w.r.t. ${\mathbf{w}}(\mathbf{s})$. As a consequence, the optimal beamforming structure can be derived by leveraging the CoV, as presented in the following proposition. 
\begin{proposition} \label{optimal_structure_prop}
    \normalfont
    \emph{(Optimal Beamforming Condition)} Given $\mathbf{g}(\mathbf{s})$, $\mathbf{U}$ and $\beta$, the optimal ${\mathbf{w}}(\mathbf{s})$ that minimizes the functional $f_c$ must satisfy the following equation: 
    \begin{align} \label{optimal_structure}
      \int_{\mathcal{S}_{\mathrm{T}}} \! \!\! c_{\mathrm{T}} (\mathbf{s}\!-\!\mathbf{z})  {\mathbf{w}}(\mathbf{z})d\mathbf{z}\!=\!\beta\mathbf{g}(\mathbf{s}) \mathbf{U} \left(\mathbf{I}\!-\!\int_{\mathcal{S}_{\mathrm{T}}} \! \! \!\mathbf{g}^\mathrm{H}(\mathbf{z})   {\mathbf{w}}  (\mathbf{z})d\mathbf{z}  \right).
\end{align}
\end{proposition}

\begin{IEEEproof}
    Please refer to Appendix \ref{app_optimal_structure_prop}.
\end{IEEEproof}
Although the optimality condition for ${\mathbf{w}}(\mathbf{s})$ is derived in (\ref{optimal_structure}), solving such an equation is quite challenging. Specifically, (\ref{optimal_structure}) is a Fredholm integral equation, where the functional variable $\mathbf{w}(\mathbf{s})$ is coupled across different integral terms. Nevertheless, the following proposition provides a closed-form solution for ${\mathbf{w}}(\mathbf{s})$ that satisfies (\ref{optimal_structure}). 
\begin{proposition} \label{closed_form_prop}
    \normalfont
    \emph{(Closed-form Solution of Optimal Beamfomring)} 
    Define $c_{\mathrm{T}}^{-1} (\mathbf{s}'-\mathbf{s})$ as the inverse of the coupling kernel $c_{\mathrm{T}} (\mathbf{s}-\mathbf{z})$, which satisfies the following identity: 
\begin{align}
    \int_{\mathcal{S}_{\mathrm{T}}} c_{\mathrm{T}}^{-1} (\mathbf{s}'-\mathbf{s})c_{\mathrm{T}} (\mathbf{s}-\mathbf{z})  d\mathbf{s} = \delta \left(\mathbf{s}'-\mathbf{z}\right).
\end{align}
Given the inverse kernel $c_{\mathrm{T}}^{-1} (\mathbf{s}'-\mathbf{s})$, the closed-form solution of ${\mathbf{w}}(\mathbf{s}')$ can be derived as 
    \begin{align} \label{closed_form}
     {\mathbf{w}}(\mathbf{s}') &=    \widetilde{\mathbf{g}}(\mathbf{s}') \mathbf{U} \left(\frac{1}{\beta}\mathbf{I}+\widetilde{\mathbf{G}} \mathbf{U}\right)^{-1},
\end{align}
where
\begin{subequations}  \label{g_G_tilde}
 \begin{align}
    \widetilde{\mathbf{g}}(\mathbf{s}')&=\int_{\mathcal{S}_{\mathrm{T}}} c_{\mathrm{T}}^{-1} (\mathbf{s}'-\mathbf{s}) \mathbf{g}(\mathbf{s})d\mathbf{s},\\
    \widetilde{\mathbf{G}}    &=\int_{\mathcal{S}_{\mathrm{T}}}\int_{\mathcal{S}_{\mathrm{T}}}\mathbf{g}^\mathrm{H}(\mathbf{s}')c_{\mathrm{T}}^{-1} (\mathbf{s}'-\mathbf{s}) \mathbf{g}(\mathbf{s}) d\mathbf{s}'d\mathbf{s}.
\end{align}   
\end{subequations}
\end{proposition}

\begin{IEEEproof}
    Please refer to Appendix \ref{app_closed_form_prop}.
\end{IEEEproof}

Although ${\mathbf{w}}(\mathbf{s}')$ can be determined via (\ref{closed_form}), the calculation of $\widetilde{\mathbf{g}}(\mathbf{s}')$ and $\widetilde{\mathbf{G}}$ relies on the inverse kernel $c_{\mathrm{T}}^{-1} (\mathbf{s}'-\mathbf{s})$. To this end, we propose the KA method to determine $c_{\mathrm{T}}^{-1} (\mathbf{s}'-\mathbf{s})$. Recalling the expression of ${c}_\mathrm{T}(\mathbf{s})$ in (\ref{mutual_kernel}), it is comprised of the dissipated mutual coupling kernel ${c}_\mathrm{diss}(\mathbf{s})$ and the radiated mutual coupling kernel ${c}_\mathrm{rad}(\mathbf{s})$. Notably, ${c}_\mathrm{rad}(\mathbf{s})$ involves a complex second-order derivative, which complicates the determination of $c_{\mathrm{T}}^{-1} (\mathbf{s}'-\mathbf{s})$.  To address this, an approximation of ${c}_\mathrm{rad}(\mathbf{s})$ with a simplified form is derived based on a wavenumber-domain approach, thereby facilitating the calculation of $c_{\mathrm{T}}^{-1} (\mathbf{s}'-\mathbf{s})$.  

Before proceeding with the derivation, we present the two-dimensional Fourier transform for a function $f(\mathbf{s})$ and the inverse transform, which are respectively given by 
\begin{subequations}\label{FT}
   \begin{align}
    F(\boldsymbol{\kappa})&=\mathcal{F}\{f\}(\boldsymbol{\kappa})=\iint_{-\infty}^{+\infty} f(\mathbf{s})e^{-\jmath \boldsymbol{\kappa}^\mathrm{T}\mathbf{s}}ds_{x} ds_{y},\\
    f(\mathbf{s})&=\frac{1}{(2\pi)^2}\iint_{-\infty}^{+\infty} F(\boldsymbol{\kappa})e^{\jmath \boldsymbol{\kappa}^\mathrm{T}\mathbf{s}}d\kappa_{x} d\kappa_{y}.
\end{align} 
\end{subequations}
Here,  $F(\boldsymbol{\kappa})$ is the wavenumber-domain representation of $f(\mathbf{s})$, where $\boldsymbol{\kappa}=[\kappa_x,\kappa_y,0]^\mathrm{T}$ is the variable in wavenumber-domain. Leveraging the linearity and differentiation properties of the Fourier transform, the wavenumber-domain representation of ${c}_\mathrm{rad}(\mathbf{s})$ can be expressed as 
\begin{align} \label{crad_FT}
    {C}_\mathrm{rad}(\boldsymbol{\kappa})
    &= \kappa_0 Z_0 \left( \mathcal{F}\{\phi(\mathbf{s})\}-\frac{\kappa_y^2}{\kappa_0^2}  \mathcal{F}\{\phi(\mathbf{s})\}\right), 
\end{align}
where the Fourier transform of $\phi(\mathbf{s})$ is provided in the following proposition. 
\begin{proposition} \label{FT_prop}
    \normalfont
    \emph{(Fourier Transform of $\phi(\mathbf{s})$)} Based on the Weyl identity and the Euler's formula, the Fourier transform of $\phi(\mathbf{s})$ can be obtained as
    \begin{align} \label{phi_FT}
    \mathcal{F}\{\phi(\mathbf{s})\} = \begin{cases}
        \frac{1}{2\sqrt{\kappa_0^2-\|\boldsymbol{\kappa}\|^2}},& \|\boldsymbol{\kappa}\| \leq \kappa_0, \\
        0,& \|\boldsymbol{\kappa}\| \geq \kappa_0.
    \end{cases}
\end{align}
\end{proposition}
\begin{IEEEproof}
    Please refer to Appendix \ref{app_FT_prop}.
\end{IEEEproof}

By substituting the Fourier transform $\mathcal{F}\{\phi(\mathbf{s})\}$ in (\ref{phi_FT}) into (\ref{crad_FT}), ${C}_\mathrm{rad}(\boldsymbol{\kappa})$ is obtained as \begin{align}
    {C}_\mathrm{rad}(\boldsymbol{\kappa})
    &= \begin{cases}
        \frac{  Z_0 ( 1- {\kappa_y^2}/{\kappa_0^2})}{2\sqrt{1-\|\boldsymbol{\kappa}\|^2/\kappa_0^2}},  & \|\boldsymbol{\kappa}\| \leq \kappa_0,\\
        0,  & \|\boldsymbol{\kappa}\| \geq \kappa_0.
    \end{cases}
\end{align}
Therefore, by applying the inverse transform in (\ref{FT}{b}), ${c}_\mathrm{rad}(\mathbf{s})$ can be rewritten as
\begin{align} \label{crad_IFT}
    {c}_\mathrm{rad}(\mathbf{s})&=\frac{1}{(2\pi)^2}\iint_{\|\boldsymbol{\kappa}\| \leq \kappa_0} {C}_\mathrm{rad}(\boldsymbol{\kappa})e^{\jmath \boldsymbol{\kappa}^\mathrm{T}\mathbf{s}}d\boldsymbol{\kappa} \notag \\
    &=\frac{1}{(2\pi)^2} \! \int_{-\kappa_0}^{\kappa_0}\! \int_{-\sqrt{\kappa_0^2-\kappa_x^2}}^{\sqrt{\kappa_0^2-\kappa_x^2}} \! {C}_\mathrm{rad}(\boldsymbol{\kappa})e^{\jmath \boldsymbol{\kappa}^\mathrm{T}\mathbf{s}}d {\kappa_y} d {\kappa_x}.
\end{align}

To obtain a tractable expression of ${c}_\mathrm{rad}(\mathbf{s})$, Gauss-Legendre quadrature is employed to approximate the continuous integral in (\ref{crad_IFT}). Specifically, the Gauss-Legendre quadrature states that the integral of a function $g(x)$ over the interval $(a,b)$ can be approximated by a weighted sum of multiple Gauss-Legendre polynomial terms, i.e., 
\begin{align} \label{GL}
    \int_{a}^{b} g(x) dx\approx \frac{b-a}{2} \sum_{m=1}^{M} \omega_m g\left(\frac{b-a}{2} \theta_m + \frac{a+b}{2}\right),
\end{align}
where $M$ is the order of approximation, $\theta_m$ denote the root of the $m$-th Gauss-Legendre polynomial, and $\omega_m$ represents the corresponding weights. As $M$ increases, the Gauss-Legendre quadrature converges geometrically, where a moderate $M$ can be typically sufficient to yield  near-exact numerical integration results.   

Based on (\ref{GL}) and (\ref{crad_IFT}), for any $m,m' \in \{1,\cdots,M\}$, we denote $\left(\kappa_m^{(x)}, W_m^{(x)}\right)$ and $\left(\kappa_{m,m'}^{(y)},W_{m,m'}^{(y)}\right)$ as the
Gauss-Legendre coefficient pairs for the integrals, given by 
\begin{subequations}
\begin{align}
    \kappa_m^{(x)} = \kappa_0 \theta_m, ~~~ \kappa_{m,m'}^{(y)} = \sqrt{\kappa_0^2-\left(\kappa_{m}^{(x)}\right)^2}\theta_{m'}, \\
    W_m^{(x)}=\kappa_0 \omega_m,~~~ W_{m,m'}^{(y)} = \sqrt{\kappa_0^2-\left(\kappa_{m}^{(x)}\right)^2}\omega_{m'}.
\end{align}    
\end{subequations}
Furthermore, we define
\begin{subequations}
\begin{align}
    \widetilde{\rho}_{m,m'} &= \frac{W_m^{(x)}W_{m,m'}^{(y)}}{(2\pi)^2}{C}_\mathrm{rad}(\widetilde{\boldsymbol{\kappa}}_{m,m'}), \\
    \widetilde{\boldsymbol{\kappa}}_{m,m'} &= \left[\kappa_m^{(x)}, \kappa_{m,m'}^{(y)},0 \right]^\mathrm{T},
\end{align}    
\end{subequations}
so that ${c}_\mathrm{rad}(\mathbf{s})$ in (\ref{crad_IFT}) can be approximated and expressed in the following simplified form: 
\begin{align} \label{crad_approx}
    {c}_\mathrm{rad}(\mathbf{s}) &\approx \sum_{m=1}^{M} \sum_{m'=1}^{M} \widetilde{\rho}_{m,m'} e^{\jmath \widetilde{\boldsymbol{\kappa}}_{m,m'}^\mathrm{T}\mathbf{s}} \overset{(\text{a})}{=} \sum_{i=1}^{I}  {\rho}_{i} e^{\jmath {\boldsymbol{\kappa}}_{i}^\mathrm{T}\mathbf{s}}.
\end{align}
In step (a), we re-index the weights and wavenumber vectors to further streamline the expression, with $I = M^2$. Given the approximation of the radiated kernel in (\ref{crad_approx}), the overall coupling kernel can be expressed as 
\begin{align} \label{ct_approx}
    {c}_\mathrm{T}(\mathbf{s}-\mathbf{z}) 
    &= Z_s \delta(\mathbf{s}-\mathbf{z}) + \sum_{i=1}^{I}  {\rho}_{i} e^{\jmath {\boldsymbol{\kappa}}_{i}^\mathrm{T}\mathbf{s}}e^{-\jmath {\boldsymbol{\kappa}}_{i}^\mathrm{T}\mathbf{z}}.
\end{align}
With the approximated expression of ${c}_\mathrm{T}(\mathbf{s}-\mathbf{z})$ in (\ref{ct_approx}), the inverse of ${c}_\mathrm{T}(\mathbf{s}-\mathbf{z})$ becomes tractable, which is given in the  following proposition.
\begin{proposition} \label{inverse_prop}
    \normalfont
    \emph{(Inverse of Coupling Kernel)} The inverse of ${c}_\mathrm{T}(\mathbf{s}-\mathbf{z})$ can be given by 
    \begin{align} \label{inverse_ct}
    \!{c}_\mathrm{T}^{-1}(\mathbf{s}'\!-\!\mathbf{s})\!=\! \frac{1}{Z_s}\delta(\mathbf{s}'\!-\!\mathbf{s}) \!-\! \!\sum_{i\!=\!1}^{I} \sum_{i'\!=\!1}^{I} \!\frac{{\rho}_{i'} d_{i,i'}}{Z_s^2} e^{\jmath {\boldsymbol{\kappa}}_{i}^\mathrm{T}\mathbf{s}'}e^{-\jmath {\boldsymbol{\kappa}}_{i'}^\mathrm{T}\mathbf{s}},
\end{align}
where $d_{i,i'}$ is the ($i$, $i'$)-th entry of the matrix $\mathbf{D}=(\mathbf{I}+\boldsymbol{\Lambda}\mathbf{Q})^{-1}$. Here, $\boldsymbol{\Lambda}$ is a diagonal matrix, which is given by 
\begin{align}
    \boldsymbol{\Lambda}=\mathrm{diag}\left\{\frac{\rho_1}{Z_s},\cdots, \frac{\rho_I}{Z_s}\right\},
\end{align}
while the entries in $\mathbf{Q}$ are calculated by
\begin{align}    \left[\mathbf{Q}\right]_{i,i'}&=\int_{\mathcal{S}}e^{\jmath  {\boldsymbol{\kappa}}_{i'}^\mathrm{T}\mathbf{s}}e^{-\jmath  {\boldsymbol{\kappa}}_{i}^\mathrm{T}\mathbf{s}} d\mathbf{s} \notag \\
    &=\int_{-\frac{L_x}{2}}^{\frac{L_x}{2}}e^{-\jmath  \Delta \kappa_{i,i'}^{(x)} {s}_x} d s_x \int_{-\frac{L_y}{2}}^{\frac{L_y}{2}} e^{-\jmath  \Delta \kappa_{i,i'}^{(y)} {s}_y} d s_y \notag \\
    &=L_x L_y \mathrm{sinc}\left(\frac{L_x \Delta \kappa_{i,i'}^{(x)}}{2}\right) \mathrm{sinc}\left(\frac{L_y \Delta \kappa_{i,i'}^{(y)}}{2}\right),
\end{align}
where we denote $\Delta \kappa_{i,i'}^{(x)}$ and $\Delta \kappa_{i,i'}^{(y)}$ as the $x$ and $y$ components of $({\boldsymbol{\kappa}}_{i}-{\boldsymbol{\kappa}}_{i'} )$, respectively.
\end{proposition}
\begin{IEEEproof}
    The derivation details are similar to the Appendix B in \cite{Optimal_CAPA}, which are thus omitted here for conciseness.
\end{IEEEproof}

Furthermore, by substituting ${c}_\mathrm{T}^{-1}(\mathbf{s}'-\mathbf{s})$ in (\ref{inverse_ct}) to (\ref{g_G_tilde}{a}) and (\ref{g_G_tilde}{b}), $\widetilde{\mathbf{g}}(\mathbf{s}')$ and $\widetilde{\mathbf{G}} $ can be simplified as follows. Specifically, $\widetilde{\mathbf{g}}(\mathbf{s}')$ can be rewritten as 
\begin{align} \label{g_tilde}
    \widetilde{\mathbf{g}}(\mathbf{s}') 
    &= \frac{1}{Z_s}\int_{\mathcal{S}_{\mathrm{T}}}\delta(\mathbf{s}'-\mathbf{s}) \mathbf{g}(\mathbf{s}) d\mathbf{s} \notag\\
    &~~- \sum_{i=1}^{I} \sum_{i'=1}^{I} \frac{{\rho}_{i'} d_{i,i'}}{Z_s^2} \int_{\mathcal{S}_{\mathrm{T}}}e^{\jmath {\boldsymbol{\kappa}}_{i}^\mathrm{T}\mathbf{s}'}e^{-\jmath {\boldsymbol{\kappa}}_{i'}^\mathrm{T}\mathbf{s}}\mathbf{g}(\mathbf{s})d\mathbf{s} \notag\\
    &=  \!\frac{1}{Z_s} \mathbf{g}(\mathbf{s}')  \!-\! \!\sum_{i=1}^{I} \sum_{i'=1}^{I} \frac{{\rho}_{i'} d_{i,i'}}{Z_s^2} \!e^{\jmath {\boldsymbol{\kappa}}_{i}^\mathrm{T}\mathbf{s}'}\!\!\!\int_{\mathcal{S}_{\mathrm{T}}}\!\!\!e^{-\jmath {\boldsymbol{\kappa}}_{i'}^\mathrm{T}\mathbf{s}}\mathbf{g}(\mathbf{s})d\mathbf{s}\notag\\
     &=  \frac{1}{Z_s} \mathbf{g}(\mathbf{s}')  - \frac{1}{Z_s} \sum_{i=1}^{I}   e^{\jmath {\boldsymbol{\kappa}}_{i}^\mathrm{T}\mathbf{s}'} \mathbf{b}_i.
\end{align}
Here, $\mathbf{b}_i$ is the $i$-th row of $\mathbf{B}=\mathbf{D}\boldsymbol{\Lambda} \mathbf{A}$, where the $i'$-th row of  $\mathbf{A}$ can be calculated by
\begin{align} \label{a}
    \mathbf{a}_{i'}=\int_{\mathcal{S}_{\mathrm{T}}}e^{-\jmath {\boldsymbol{\kappa}}_{i'}^\mathrm{T}\mathbf{s}}\mathbf{g}(\mathbf{s})d\mathbf{s}. 
\end{align}
By denoting $\mathbf{G}=\int_{\mathcal{S}_{\mathrm{T}}} \mathbf{g}^\mathrm{H}(\mathbf{s}) \mathbf{g}(\mathbf{s}) d\mathbf{s}$,  $\widetilde{\mathbf{G}} $ can be obtained by 
\begin{align} \label{G_tilde}
    \widetilde{\mathbf{G}}    
    &=\frac{1}{Z_s}\int_{\mathcal{S}_{\mathrm{T}}} \mathbf{g}^\mathrm{H}(\mathbf{s}) \mathbf{g}(\mathbf{s}) d\mathbf{s}\notag\\
    &~~ -\!\sum_{i=1}^{I} \sum_{i'=1}^{I}\!\frac{{\rho}_{i'} d_{i,i'}}{Z_s^2}\!\!\int_{\mathcal{S}_{\mathrm{T}}}\!\!\mathbf{g}^\mathrm{H}(\mathbf{s}')  e^{\jmath {\boldsymbol{\kappa}}_{i}^\mathrm{T}\mathbf{s}'} d\mathbf{s}'\!\! \int_{\mathcal{S}_{\mathrm{T}}}\!\!\!e^{-\jmath {\boldsymbol{\kappa}}_{i'}^\mathrm{T}\mathbf{s}} \mathbf{g}(\mathbf{s}) d\mathbf{s} \notag \\
    &= \frac{1}{Z_s}\mathbf{G}-\frac{1}{Z_s} \mathbf{A}^\mathrm{H}\mathbf{D}\boldsymbol{\Lambda} \mathbf{A}.
\end{align}

Building upon the aforementioned results, the optimal ${\mathbf{w}}(\mathbf{s})$ can be obtained by iteratively updating ${v}_k$, $\mu_k$ and ${\mathbf{w}}(\mathbf{s})$, where the overall algorithm for solving problem (\ref{problem_1}) is summarized in Algorithm \ref{algo_KA_WMMSE}. 

\begin{algorithm}[!t] 
\caption{Proposed KA-based WMMSE Algorithm} 
    \begin{algorithmic}[1]  \label{algo_KA_WMMSE}
        \REQUIRE ~  CAPA system parameters;
        \STATE Initialize $\mathbf{w}(\mathbf{s})$;
         \REPEAT  
         \STATE Calculate $\mathbf{e}_k$ and $\widetilde{\sigma}_k^2$ to update ${v}_k^{\mathrm{MMSE}}$ using (\ref{MMSE_receiver}) and $\varepsilon_k^\mathrm{MMSE}$ using (\ref{closed_form});
         \STATE Construct $\mathbf{U}$ with $\varepsilon_k^\mathrm{MMSE}$ and calculate $\beta$ by (\ref{g_U_beta}{c});
         \STATE Obtain $\boldsymbol{\Lambda}$ and $\mathbf{D}$ in the Proposition \ref{inverse_prop}, and calculate $\widetilde{\mathbf{g}}(\mathbf{s}')$  and $\widetilde{\mathbf{G}}$ using (\ref{g_tilde}) and (\ref{G_tilde}), respectively. 
        \STATE Update $\mathbf{w}(\mathbf{s})$ through Proposition \ref{closed_form_prop};
        \UNTIL convergence;
        \STATE Scale $\mathbf{w}(\mathbf{s})$ by  (\ref{scale_w})
    \end{algorithmic}
\end{algorithm}
 
\subsection{Matrix Implementation}
In Algorithm \ref{algo_KA_WMMSE},  iterative computations of integrals to update ${v}_k$, $\mathbf{U}$ and $\mathbf{w}(\mathbf{s})$ are required, posing a significant computational burden for practical implementations.  To bypass these intensive integral computations and facilitate the practical implementations, we leverage Gauss-Legendre quadrature presented in (\ref{GL}) to develop a matrix-based implementation scheme for Algorithm \ref{algo_KA_WMMSE}.  

First, we focus on ${v}_k$ in  (\ref{MMSE_receiver}) and $\varepsilon_k^\mathrm{MMSE}$ in (\ref{MMSE}), where the calculation of $\mathbf{e}_k$ and $\widetilde{\sigma}_k$ involves integral evaluations. By applying the Gauss-Legendre quadrature, the $k$-th entry of $\mathbf{e}_k$ can be calculated by 
\begin{align}
    e_{k,k} &= \int_{\mathcal{S}_{\mathrm{T}}} h_k(\mathbf{s})  {w}_k(\mathbf{s}) d \mathbf{s} \notag \\
    &\approx \sum_{m=1}^{M} \sum_{m'=1}^{M}  \frac{\omega_m \omega_{m'} A_\mathrm{T}}{4} h_k(\mathbf{s}_{m,m'}) {w}_k(\mathbf{s}_{m,m'}) \notag \\
    &= \mathbf{h}_k \boldsymbol{\Phi}_\mathrm{T} {\mathbf{w}}_k,
\end{align}
where we define
 \begin{align}
    {\mathbf{w}}_k &= \left[{{w}}_k(\mathbf{s}_{1,1}),\cdots,  {{w}}_k(\mathbf{s}_{M,M})\right]^\mathrm{T} \in \mathbb{C}^{I\times 1}, \\
    \mathbf{h}_k &= \left[h_k(\mathbf{s}_{1,1}),\cdots, h_k(\mathbf{s}_{M,M})\right]  \in \mathbb{C}^{1 \times I}, \\
    \boldsymbol{\Phi}_\mathrm{T}&=\frac{1}{4} A_\mathrm{T} \mathrm{diag}\left\{ \omega_1\omega_1,\cdots,\omega_M\omega_M\right\}\in \mathbb{C}^{I \times I}.
\end{align}   
Accordingly, $\mathbf{e}_k$ can be given by
\begin{align}       \label{ek}
    \mathbf{e}_k 
    &= \mathbf{h}_k \boldsymbol{\Phi}_\mathrm{T}  {\mathbf{W}},
\end{align}
where ${\mathbf{W}}$ is the matrix form of ${\mathbf{w}}(\mathbf{s}')$, given by
\begin{align}
    {\mathbf{W}}&= \left[ {\mathbf{w}}^\mathrm{T}(\mathbf{s}_{1,1}),\cdots,  {\mathbf{w}}^\mathrm{T}(\mathbf{s}_{M,M})\right]^\mathrm{T} \in \mathbb{C}^{I \times K}.
\end{align}
To calculate $\widetilde{\sigma}_k$, we first substitute (\ref{ct_approx}) to (\ref{expression_SINR_sigma}{b}), which yields
\begin{align}
    \widetilde{\sigma}_k  =&\frac{\sigma_k^2}{2 P_\mathrm{T}}Z_s\int_{\mathcal{S}_{\mathrm{T}}}\left\|{\mathbf{w}}(\mathbf{s})\right\|^2d \mathbf{s} \notag \\
  &+ \frac{\sigma_k^2}{2 P_\mathrm{T}}\!\int_{\mathcal{S}_{\mathrm{T}}} \!\int_{\mathcal{S}_{\mathrm{T}}} \!\! \!{\mathbf{w}}(\mathbf{s}) \!\left( \sum_{i=1}^{I}  {\rho}_{i} e^{\jmath {\boldsymbol{\kappa}}_{i}^\mathrm{T}\mathbf{s}}e^{-\jmath {\boldsymbol{\kappa}}_{i}^\mathrm{T}\mathbf{z}} \right) \!{\mathbf{w}}^\mathrm{H}(\mathbf{z}) d \mathbf{z} d \mathbf{s}.
\end{align}
By employing the Gauss-Legendre quadrature,  $\widetilde{\sigma}_k$ can then be calculated by  
\begin{align} \label{noise_matrix}
     \!\!\widetilde{\sigma}_k \! =&\frac{\sigma_k^2 Z_s}{2 P_\mathrm{T}}\!\left[\mathrm{tr}\!\left({\mathbf{W}}^\mathrm{H}\boldsymbol{\Phi}_\mathrm{T} {\mathbf{W}}\right) \!+\!\mathrm{tr}\!\left({\mathbf{W}}^\mathrm{H}\boldsymbol{\Phi}_\mathrm{T}  \mathbf{X}^\mathrm{H}\boldsymbol{\Lambda} \mathbf{X}\boldsymbol{\Phi}_\mathrm{T} {\mathbf{W}}\right)\right]\!,
\end{align}
with
\begin{align}
\mathbf{X}&= \left[\mathbf{x}_{1}^\mathrm{T},\cdots,\mathbf{x}_{I}^\mathrm{T}\right]^\mathrm{T}\in \mathbb{C}^{I \times I},\\
    \mathbf{x}_{i} &= \left[e^{-\jmath {\boldsymbol{\kappa}}_{i}^\mathrm{T}\mathbf{s}_{1,1}},\cdots, e^{-\jmath {\boldsymbol{\kappa}}_{i}^\mathrm{T}\mathbf{s}_{M,M}}\right]  \in \mathbb{C}^{1 \times I}.  
\end{align}
Once $\mathbf{e}_k$ and $\widetilde{\sigma}_k$ are obtained, ${v}_k$ and $\varepsilon_k^\mathrm{MMSE}$ are calculated via (\ref{MMSE_receiver}) and  (\ref{MMSE}), respectively. Subsequently, the MSE weights  $\mu_k=(\varepsilon_k^\mathrm{MMSE})^{-1}$ and the weight matrix $\mathbf{U}=\mathrm{diag}(\mu_1,\cdots,\mu_K)$ can be directly updated. With the given ${v}_k$ and $\mu_k$, the coefficient $\beta$ in (\ref{g_U_beta}{c}) can also be explicitly determined as follows:
\begin{align}
    \frac{1}{\beta}&=\sum_{k=1}^{K} \mu_k \frac{\sigma_k^2}{2 P_\mathrm{T}}|{v}_k|^2 = \frac{1}{2 P_\mathrm{T}} \mathrm{tr}(\mathbf{V}^\mathrm{H} \mathbf{U}\boldsymbol{\Sigma}\mathbf{V}),
\end{align}
where
\begin{align}
    \mathbf{V} &= \mathrm{diag}(v_1,\cdots,v_K),~~\boldsymbol{\Sigma} = \mathrm{diag}(\sigma_1^2,\cdots,\sigma_K^2).
\end{align}

In the subsequence, the optimal beamformer $\mathbf{W}$ is updated as follows. To compute ${\mathbf{W}}$ in \eqref{closed_form}, the matrix-based representations of $ \widetilde{\mathbf{g}}(\mathbf{s}')$ and $ \widetilde{\mathbf{G}}$  are required, where $\mathbf{G}$ and $\mathbf{A}$ involve integral evaluations. To this end, we first define 
\begin{align}
    \mathbf{H}&=\left[\mathbf{h}_1^\mathrm{T},\cdots, \mathbf{h}_K^\mathrm{T}\right]^\mathrm{T} \in \mathbb{C}^{K \times I},
\end{align}
based on which the matrix form of $\mathbf{g}(\mathbf{s})$ can be given by
\begin{align} \label{g_matrix}
     \!\!\!\!\!\begin{bmatrix}
        \!\mathbf{g}(\mathbf{s}_{1,1})\!\\
        \vdots\\
        \mathbf{g}(\mathbf{s}_{M,M})\!
    \end{bmatrix}\!\!=\!\!\begin{bmatrix}
      \!\! {v}_1 h_1^*  ( \mathbf{s}_{1,1}) &\!\!\!\cdots&\!\!\!\!{v}_K h_K^*  ( \mathbf{s}_{1,1})\! \\
        &\!\!\!\ddots\!&\\
       \! {v}_1 h_1^*  ( \mathbf{s}_{M,M}) &\!\!\!\cdots\!&\!\!{v}_K h_K^*  ( \mathbf{s}_{M,M})\!
    \end{bmatrix}\!\!=\!\mathbf{H}^\mathrm{H} \mathbf{V}.
\end{align}
Given the expression of (\ref{g_matrix}), $\mathbf{G}$ can be approximated as follows:
\begin{align}
    \mathbf{G}& = \int_{\mathcal{S}_{\mathrm{T}}} \!\! \mathbf{g}^\mathrm{H}(\mathbf{s}) \mathbf{g}(\mathbf{s}) d\mathbf{s} \notag \\
    & \approx \sum_{m=1}^{M} \!\sum_{m'=1}^{M}\!\! \frac{\omega_m \omega_{m'}A_\mathrm{T}}{4}\mathbf{g}^\mathrm{H}(\mathbf{s}_{m,m'}) \mathbf{g}(\mathbf{s}_{m,m'}) \notag \\
    &= \mathbf{V}^\mathrm{H}  \mathbf{H} \boldsymbol{\Phi}_\mathrm{T}\mathbf{H}^\mathrm{H} \mathbf{V}.
\end{align}
To compute $\mathbf{A}$, we should be obtained $\mathbf{a}_{i'}$ in (\ref{a}) first, which can be compactly approximated by
\begin{align}
    \mathbf{a}_{i'}&=\int_{\mathcal{S}_{\mathrm{T}}}e^{-\jmath {\boldsymbol{\kappa}}_{i'}^\mathrm{T}\mathbf{s}}\mathbf{g}(\mathbf{s})d\mathbf{s}  \notag\\
    &\approx \sum_{m=1}^{M} \sum_{m'=1}^{M} \frac{\omega_m \omega_{m'}A_\mathrm{T}}{4}e^{-\jmath {\boldsymbol{\kappa}}_{i'}^\mathrm{T}\mathbf{s}_{m,m'}}\mathbf{g}(\mathbf{s}_{m,m'})\notag \\
    &=\mathbf{x}_{i'}\boldsymbol{\Phi}_\mathrm{T}\mathbf{H}^\mathrm{H} \mathbf{V}.
\end{align}
Accordingly, $\mathbf{A}$ incorporates $\mathbf{a}_{i'}$ in $i$-th row, thus $\mathbf{A}$ can be rewritten as 
\begin{align}
    \mathbf{A} = \mathbf{X}\boldsymbol{\Phi}_\mathrm{T}\mathbf{H}^\mathrm{H} \mathbf{V}.
\end{align} 

When $\mathbf{G}$ and $\mathbf{A}$ are computed, $ \widetilde{\mathbf{G}}$ can be obtained by (\ref{G_tilde}). Besides, $\mathbf{B}$ can be given by $\mathbf{B}=\mathbf{D}\boldsymbol{\Lambda} \mathbf{A}=\mathbf{D}\boldsymbol{\Lambda} \mathbf{X}\boldsymbol{\Phi}_\mathrm{T}\mathbf{H}^\mathrm{H} \mathbf{V}$.
With the calculated $\mathbf{B}$, the matrix form of $\widetilde{\mathbf{g}}(\mathbf{s}')$ can be expressed as
\begin{align}
    \!\!{\boldsymbol{\Gamma}} \!=\! \left[
        \widetilde{\mathbf{g}}^\mathrm{T}(\mathbf{s}'_{1,1}),
        \cdots,
        \widetilde{\mathbf{g}}^\mathrm{T}(\mathbf{s}'_{M,M})
    \right]^\mathrm{T}  = \frac{1}{Z_s} \mathbf{H}^\mathrm{H} \mathbf{V}\!-\!\frac{1}{Z_s}  \mathbf{X}^\mathrm{H} \mathbf{B}.
\end{align}
Based on the above results,  $ {\mathbf{W}}$ can be updated by 
\begin{align} \label{optimal_W}
    {\mathbf{W}} = {\boldsymbol{\Gamma}} \mathbf{U} \left(\frac{1}{\beta}\mathbf{I}+\widetilde{\mathbf{G}} \mathbf{U}\right)^{-1}.
\end{align}
To facilitate the matrix-based implementation, each step in Algorithm \ref{algo_KA_WMMSE} that originally requires integral evaluations can be replaced by the corresponding matrix forms derived above. This transformation effectively converts the continuous optimization into a sequence of computationally efficient linear algebraic operations.

\subsection{Convergence and Computational Complexity Analysis}
\subsubsection{Convergence}
While the proposed algorithm is based on the WMMSE framework, its convergence proof can be treated as a specific case of the results in \cite{WMMSE}. Specifically, the process of our algorithm is to minimize the following function w.r.t $v_k$, $\mu_k$, and $\mathbf w(\mathbf{s})$, i.e.,
\begin{align} \label{augment}
\mathcal{L}({\mathbf w}(\mathbf s),v_k,\mu_k)
= \sum_{k=1}^K \left(\mu_k \varepsilon_k - \log \mu_k\right).
\end{align}
As shown in \cite{WMMSE}, minimizing $\mathcal{L}$ is equivalent to maximizing the sum-rate when $v_k=v_k^{\mathrm{MMSE}}$ in (\ref{MMSE_receiver}) and $\mu_k=(\varepsilon_k^{\mathrm{MMSE}})^{-1}$ in (\ref{MMSE}) are selected.  To minimize $\mathcal{L}({\mathbf w}(\mathbf s),v_k,\mu_k)$, the variables of ${\mathbf w}(\mathbf s),v_k,\mu_k$ are updated iteratively while keeping the others fixed. Given ${\mathbf w}(\mathbf s)$ and $\mu_k$, updating $v_k$ in $\mathcal{L} $ reduces to $\min \mathcal{L}( v_k)=\varepsilon_k$, where the solution is given by (\ref{MMSE_receiver}) and guarantees $\mathcal{L}({\mathbf w}^i(\mathbf s),v_k^{(i+1)},\mu_k^i)\leq \mathcal{L}({\mathbf w}^i(\mathbf s),v_k^i,\mu_k^i)$. Similarly, for fixed ${\mathbf w}(\mathbf s)$ and $v_k$, the optimal weight $\mu_k=(\varepsilon_k^{\mathrm{MMSE}})^{-1}$ minimizes $\mathcal{L}$, guaranteeing $\mathcal{L}({\mathbf w}^i(\mathbf s),v_k^{(i+1)},\mu_k^{(i+1)})\leq \mathcal{L}({\mathbf w}^i(\mathbf s),v_k^{(i+1)},\mu_k^i)$.  Finally, for fixed $v_k$ and $\mu_k$, (\ref{augment}) is simplified to the formulation in  (\ref{WMMSE_formulation_w}). The CoV approach employed in Proposition \ref{optimal_structure_prop} minimizes (\ref{WMMSE_formulation_w}) by optimizing ${\mathbf w}(\mathbf s)$, ensuring $\mathcal{L}({\mathbf w}^{(i+1)}(\mathbf s),v_k^{(i+1)},\mu_k^{(i+1)})\leq \mathcal{L}({\mathbf w}^i(\mathbf s),v_k^{(i+1)},\mu_k^{(i+1)})$. Consequently, as each iteration step non-increases the value of $\mathcal{L}$, the sequence $\{\mathcal{L}^{(i)}\}$ is guaranteed to be non-increasing. Moreover, $\mathcal{L}$ is lower-bounded under the finite transmit power constraint. Therefore, the convergence of Algorithm \ref{algo_KA_WMMSE} is guaranteed. 

\subsubsection{Computational Complexity}
The computational complexity of the proposed algorithm is analyzed as follows. For updating $v_k$ and $\mu_k$, the evaluations of $\mathbf{e}_k$ and $\widetilde{\sigma}_k^2$ via (\ref{ek}) and (\ref{noise_matrix}) incur a computational complexity of $\mathcal{O}(KI^2+K^2I)$. Given $v_k$ and $\mu_k$, the calculation of the beamformer $\mathbf{w}(\mathbf{s})$ in (\ref{optimal_W}) involves matrix multiplications and a matrix inversion.  Specifically, the intermediate variables ${\boldsymbol{\Gamma}}$, $\mathbf{U}$, $\beta$, $\widetilde{\mathbf{G}} $ are first computed with a total computational complexity of $\mathcal{O}(KI^2+K^2I+K^3+I^3)$. Besides, the matrix inversion requires a computational complexity of $\mathcal{O}(K^3)$. As a consequence, the overall computational complexity per iteration is dominated by $\mathcal{O}(KI^2+K^2I+K^3+I^3)$. 

\section{Extension to CAPA-to-CAPA MIMO System} \label{sec:extension}
In this section, we extend the proposed method to a CAPA-to-CAPA MIMO system. Specifically, the receiver is equipped with a CAPA $\mathcal{S}_\mathrm{R}$ with the lengths along the x- and y-axes being $\overline{L}_x$ and $\overline{L}_y$, respectively. Thus, the total aperture area of the CAPA receiver is thus given by $|\mathcal{S}_\mathrm{R}| = \overline{L}_x \cdot \overline{L}_y = A_\mathrm{R}$. Let $ {\mathbf{r}} = \left[ {r}_x,  {r}_y, 0\right]^\mathrm{T} \in \mathcal{S}_\mathrm{R}$ denote an arbitrary point on the CAPA receiver. In addition, $N$ independent data streams are assumed for signal transmission. The noisy electric field captured by the CAPA receiver can be given by 
\begin{align}
    y(\mathbf{r})&\!=\! \mathbf{u}_\mathrm{R}^\mathrm{T} \mathbf{e}_{\mathrm{rad}}(\mathbf{r}) + z(\mathbf{r})  \!= \!\int_{\mathcal{S}_{\mathrm{T}}}\!\! h(\mathbf{r}-\mathbf{s})   {j}_\mathrm{t}(\mathbf{s}) d\mathbf{s} \!+ \!z(\mathbf{r}),
\end{align}
where $z(\mathbf{r})$ represents the AWGN following a distribution of $\mathcal{CN}(0,\sigma_r^2)$. Furthermore, the channel response between  $\mathbf{s}$ and $\mathbf{r}$ retains the same form as (\ref{channel}),  which can be expressed as 
\begin{align}
    h(\mathbf{r}-\mathbf{s}) = -\jmath \kappa_0 Z_0 \left(g(\mathbf{r}-\mathbf{s})+\frac{1}{\kappa_0^2} \partial _y^2g(\mathbf{r}-\mathbf{s})\right).
\end{align}
Here, we assume that the CAPA receiver has a vertical polarization direction. Based on the above receiving model, the achievable rate of the CAPA-aided MIMO system can be formulated as 
\begin{align}
    R=\log \det \left( \mathbf{I}_N + \frac{1}{\sigma_r^2} \mathbf{Q} \right),
\end{align}
where we define $\mathbf{Q} = \int_{\mathcal{S}_{\mathrm{R}}} \mathbf{e}^{\mathrm{H}}_{\mathrm{rx}}(\mathbf{r}) \mathbf{e}_{\mathrm{rx}}(\mathbf{r}) d \mathbf{r}$ and $\mathbf{e}_{\mathrm{rx}}(\mathbf{r}) = \int_{\mathcal{S}_{\mathrm{T}}} h(\mathbf{r}-\mathbf{s}) \mathbf{w}(\mathbf{s}) d \mathbf{s}$. Considering the mutual coupling effects as the previous section, the  achievable rate maximization problem can be formulated as 
\begin{subequations} \label{problem_MIMO_1}
    \begin{align} 
        \max_{\mathbf{w}(\mathbf{s})} \quad &\log \det \left( \mathbf{I}_N + \frac{1}{\sigma_r^2} \mathbf{Q} \right) \\
        \mathrm{s.t.} \quad &  \frac{1}{2}\int_{\mathcal{S}_{\mathrm{T}}} \int_{\mathcal{S}_{\mathrm{T}}} \mathbf{\mathbf{w}(\mathbf{s}) }c_{\mathrm{T}} (\mathbf{s}-\mathbf{z}) \mathbf{w}^\mathrm{H}(\mathbf{z}) d \mathbf{z} d \mathbf{s} \le P_{\mathrm{T}}.
    \end{align}    
\end{subequations}
Consistent with Lemma \ref{equal_power_lemma}, the problem (\ref{problem_MIMO_1}) can  transformed into an equivalent unconstrained optimization problem: 
\begin{align} \label{problem_MIMO_2}
     \max_{ {\mathbf{w}}(\mathbf{s})} \quad & \widetilde{R}=\log \det \left( \mathbf{I}_N + \frac{1}{\widetilde{\sigma}_r^2} \mathbf{Q} \right),
\end{align}
where
\begin{align}
    \widetilde{\sigma}_r^2 = \frac{\sigma_r^2}{2 P_\mathrm{T}}\int_{\mathcal{S}_{\mathrm{T}}} \int_{\mathcal{S}_{\mathrm{T}}}  {\mathbf{w}}(\mathbf{s}) c_{\mathrm{T}} (\mathbf{s}-\mathbf{z}) {\mathbf{w}}^\mathrm{H}(\mathbf{z}) d \mathbf{z} d \mathbf{s}.
\end{align}
The equivalent receiving model is thus given by
\begin{align}
    \widetilde{y}(\mathbf{r})&=  \int_{\mathcal{S}_{\mathrm{T}}} h(\mathbf{r}-\mathbf{s})   {j}_\mathrm{t}(\mathbf{s}) d\mathbf{s} +  \widetilde{z}(\mathbf{r})
    =  \mathbf{e}_{\mathrm{rx}}(\mathbf{r}) \mathbf{c}+ \widetilde{ {z}}(\mathbf{r}),
\end{align}
where $\widetilde{ {z}}(\mathbf{r})\sim \mathcal{CN}(0,\widetilde{\sigma}_r^2)$ is the equivalent noise.

To solve (\ref{problem_MIMO_2}),  the WMMSE framework can be employed. In contrast to the multi-user case in Section \ref{sec:algorithm}, the symbol estimator is a set of continuous functions, i.e., $\mathbf{v}(\mathbf{r}) \in \mathbb{C}^{1\times N}$. Accordingly, the estimated data stream can be expressed as 
\begin{align}\label{c_estimate_MIMO}
    \hat{\mathbf{c}}&= \int_{\mathcal{S}_{\mathrm{R}}} \mathbf{v}^\mathrm{H}(\mathbf{r}) \widetilde{y}(\mathbf{r}) d\mathbf{r}  =\int_{\mathcal{S}_{\mathrm{R}}} \mathbf{v}^\mathrm{H}(\mathbf{r})  \mathbf{e}_{\mathrm{rx}} (\mathbf{r})d\mathbf{r} ~ \mathbf{c} + \widetilde{\mathbf{z}}(\mathbf{r}).
\end{align}
Here, $\widetilde{\mathbf{z}}(\mathbf{r})$ is the estimation noise, which is given by
\begin{align}
   \!\! \widetilde{\mathbf{z}}(\mathbf{r}) \!= \!\!\int_{\mathcal{S}_{\mathrm{R}}}\! \!\mathbf{v}^\mathrm{H}(\mathbf{r}) \widetilde{ {z}}(\mathbf{r})d\mathbf{r} \!\sim \!\mathcal{CN}\left(\mathbf{0},\widetilde{\sigma}_r^2 \!\int_{\mathcal{S}_{\mathrm{R}}} \!\!\!\mathbf{v}^\mathrm{H}(\mathbf{r}) \mathbf{v}(\mathbf{r}) d\mathbf{r} \right)\!.
\end{align}
Based on (\ref{c_estimate_MIMO}),  the MSE matrix can be calculated by 
\begin{align} \label{MSE_Matrix_MIMO}
    \mathbf{E}=&\mathbb{E}\left[(\hat{\mathbf{c}}-\mathbf{c})(\hat{\mathbf{c}}-\mathbf{c})^\mathrm{H}\right] \notag \\
    =& \left( \mathbf{I}\!-\!\int_{\mathcal{S}_{\mathrm{R}}} \!\!\mathbf{v}^\mathrm{H}(\mathbf{r})  \mathbf{e}_{\mathrm{rx}} (\mathbf{r})d\mathbf{r}\right)\!\left( \mathbf{I}\!-\!\int_{\mathcal{S}_{\mathrm{R}}} \!\!\mathbf{v}^\mathrm{H}(\mathbf{r})  \mathbf{e}_{\mathrm{rx}} (\mathbf{r})d\mathbf{r}\right)^\mathrm{H}\notag   \\
    & + \widetilde{\sigma}_r^2 \int_{\mathcal{S}_{\mathrm{R}}} \!\!\mathbf{v}^\mathrm{H}(\mathbf{r}) \mathbf{v}(\mathbf{r}) d\mathbf{r} \notag \\
    =& \left( \mathbf{I}\!-\!\int_{\mathcal{S}_{\mathrm{R}}} \!\mathbf{g}^\mathrm{H}(\mathbf{s})  {\mathbf{w}}  (\mathbf{s})d\mathbf{s}\right)\!\left( \mathbf{I}\!-\!\int_{\mathcal{S}_{\mathrm{R}}} \!\mathbf{g}^\mathrm{H}(\mathbf{s})   {\mathbf{w}}  (\mathbf{s})d\mathbf{s}\right)^\mathrm{H} \notag \\
    &+ \frac{\sigma_r^2}{2 P_\mathrm{T}}\mathbf{V} \int_{\mathcal{S}_{\mathrm{T}}} \int_{\mathcal{S}_{\mathrm{T}}}  {\mathbf{w}}(\mathbf{s}) c_{\mathrm{T}} (\mathbf{s}-\mathbf{z})  {\mathbf{w}}^\mathrm{H}(\mathbf{z}) d \mathbf{z} d \mathbf{s},
\end{align}
where we denote $\mathbf{V}=\int_{\mathcal{S}_{\mathrm{R}}} \mathbf{v}^\mathrm{H}(\mathbf{r}) \mathbf{v}(\mathbf{r}) d\mathbf{r}$ and $\mathbf{g}(\mathbf{s})=\int_{\mathcal{S}_{\mathrm{R}}} h^\mathrm{H}(\mathbf{r}-\mathbf{s}) \mathbf{v}(\mathbf{r})d\mathbf{r}$. Accordingly, the optimal $\mathbf{v}_{\mathrm{MMSE}}(\mathbf{r})$ can be obtained by
\begin{align}
    \mathbf{v}_{\mathrm{MMSE}}(\mathbf{r}) \!=\!\arg \min_{\mathbf{v}(\mathbf{r})} \mathrm{tr}(\mathbf{E}) \overset{\mathrm{(a)}}{=} \! \frac{ \mathbf{e}_{\mathrm{rx}}(\mathbf{r})}{\widetilde{\sigma}_r^2}\left(\!\mathbf{I}_N+\!\frac{1}{\widetilde{\sigma}_r^2} \mathbf{Q}\!\right)^{-1}\!\!. 
\end{align}
In step (a), CoV is employed to solve the minimization problem. Substituting $\mathbf{v}_{\mathrm{MMSE}}(\mathbf{r})$ to  (\ref{MSE_Matrix_MIMO}), the achievable minimum MSE can be given by
\begin{align}
   \mathbf{E}_\mathrm{MMSE} = \left(\mathbf{I}_N+\frac{1}{\widetilde{\sigma}_r^2} \mathbf{Q}\right)^{-1}.
\end{align}
The corresponding achievable rate of the system can be expressed as $\widetilde{R}=\log \det \left(\mathbf{E}_\mathrm{MMSE}^{-1}\right)$. Similar to the previous section, the rate maximization problem (\ref{problem_MIMO_2}) can be equivalent to the following weighted MSE minimization problem, i.e.,  
\begin{align} \label{WMMSE_problem_MIMO}
     \min_{{\mathbf{w}}(\mathbf{s}),\mathbf{v}(\mathbf{r}),\mathbf{U}} \quad & \mathrm{tr}\left( \mathbf{U} \mathbf{E}\right) - \log \det (\mathbf{U}).
\end{align}
When ${\mathbf{w}}(\mathbf{s})$ is given, $\mathbf{v}_\mathrm{MMSE}(\mathbf{r})$ and  $\mathbf{E}_\mathrm{MMSE}^{-1}$ are the optimal solutions for $\mathbf{v}(\mathbf{r})$ and $\mathbf{U}$, respectively. 

With $\mathbf{v}(\mathbf{r})$ and $\mathbf{U}$ fixed,  the problem (\ref{WMMSE_problem_MIMO}) is reduced to  
\begin{align}
     \min_{ {\mathbf{w}}(\mathbf{s}) } \quad & \mathrm{tr}\left[ \mathbf{U} \mathbf{E} ({\mathbf{w}}(\mathbf{s})) \right].
\end{align}
The optimal $ {\mathbf{w}}(\mathbf{s})$ can be obtained via CoV, which yields 
\begin{align}
     \Re\left\{\int_{\mathcal{S}_{\mathrm{T}}}  {\mathbf{p}}(\mathbf{s}) \boldsymbol{\eta}^\mathrm{H}(\mathbf{s})d \mathbf{s}\right\} =0.
\end{align}
Here, we denote 
\begin{align}
    {\mathbf{p}}(\mathbf{s})  =& \mathbf{g}(\mathbf{s}) \mathbf{U} \left(\int_{\mathcal{S}_{\mathrm{T}}}  \mathbf{g}^\mathrm{H}(\mathbf{z})  {\mathbf{w}}  (\mathbf{z})d\mathbf{z}  -\mathbf{I}_N\right) \notag \\
    &+ \frac{1}{\beta} \int_{\mathcal{S}_{\mathrm{T}}} c_{\mathrm{T}} (\mathbf{s}-\mathbf{z})  {\mathbf{w}}(\mathbf{z})d\mathbf{z}, \\
    \frac{1}{\beta} =& \frac{\sigma_r^2 \mathrm{tr}(\mathbf{U}\mathbf{V})}{2P_\mathrm{T}}.
\end{align}
According to the fundamental lemma of CoV, $ {\mathbf{p}}(\mathbf{s})=0$ must hold, leading to the optimal beamformer condition, i.e.,
\begin{align}\label{optimal_structure_MIMO}
      \int_{\mathcal{S}_{\mathrm{T}}} \! \! c_{\mathrm{T}} (\mathbf{s}\!-\!\mathbf{z})  {\mathbf{w}}(\mathbf{z})d\mathbf{z}\!=\!\beta\mathbf{g}(\mathbf{s}) \mathbf{U} \left(\!\mathbf{I}\!-\!\!\int_{\mathcal{S}_{\mathrm{T}}} \!\! \mathbf{g}^\mathrm{H}(\mathbf{z})   {\mathbf{w}}  (\mathbf{z})d\mathbf{z} \! \right)\!.
\end{align}
It can be observed that (\ref{optimal_structure_MIMO}) is mathematically identical to that of the multi-user system in (\ref{optimal_structure}). Consequently, the subsequent derivation of optimal $ {\mathbf{w}}(\mathbf{s})$ follows the same procedures as detailed in Section \ref{sec:algorithm}. For the sake of conciseness, these redundant details are omitted here. 

\section{Simulation Results} \label{simulation}
This section evaluates the performance of the proposed KA-based WMMSE algorithm through $100$ Monte Carlo trials. Unless otherwise specified, the default parameters are configured as follows. The system operates at a carrier frequency of $2.4$ GHz. The CAPA transmitter is modeled as a square surface, where its lengths along $x$- and $y$-axes are $L_x\!=\!L_y\!=\!0.5$ m. Thus, the total area of CAPA is  $A_\mathrm{T}\!=\!0.25$ m$^2$. Besides, we model the CAPA surface as copper, where the conductivity and the permeability are typically set to $\sigma_s\!=\!5.8\times 10^7$ S/m and $\mu_s\!=\!4\pi\times 10^{-7}$ H/m. The free-space impedance is set to $\eta=120\pi$ $\Omega$ and the transmit power is fixed at $P_t=1$ W.  We consider that $K\!=\!4$ users are randomly distributed within a circular zone centered at ($30$, $-30$, $50$)  m with a radius of $15$ m. Finally, the Gauss–Legendre quadrature with $M\! = \!30$ sampling point is employed to calculate the integrals. 

For comparison, the following benchmarks are considered. Under the same CAPA-aided system setup, the proposed algorithm is compared to the Fourier-based method \cite{Zhaolin_CAPA_BF}. This method approximates the continuous functions using a finite number of Fourier series terms and transforms the problem into the classical  finite-dimensional optimization problem.  In addition, the CAPA-aided system is compared to the conventional SPDA system, where the transmit beamforming is optimized using the algorithm proposed in \cite{WMMSE}. Moreover, we compare the results obtained with mutual coupling effects to those without considering mutual coupling effects, thereby highlighting the necessity and rationality of the proposed mutual coupling-aware design. 

\begin{figure}[!t]
    \centering
    \includegraphics[scale=0.58]{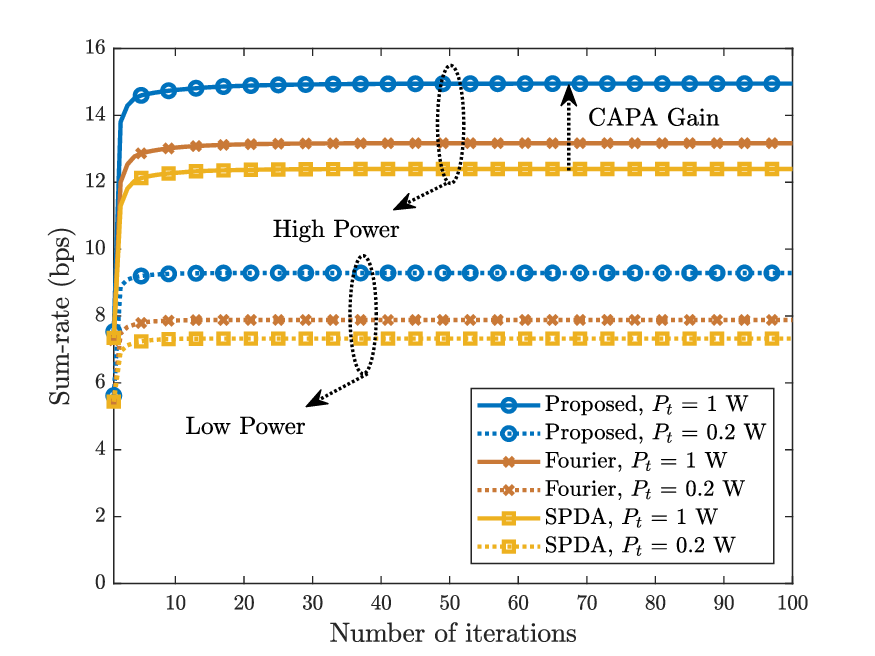}
    \caption{ Convergence behavior of the proposed KA-based WMMSE algorithm under different transmit power.}
    \label{fig:convergence}
\end{figure}


\subsection{Multi-user Scenario}
In Fig. \ref{fig:convergence}, the convergence behaviors of the proposed KA-based WMMSE algorithm are illustrated. It can be observed that both the proposed approach and the benchmarks converge under different transmit power, stabilizing within a few iterations. Moreover, as expected, schemes operating with higher transmit power consistently achieves higher sum-rate. As compared to the conventional SPDA system,  the CAPA-aided system exhibits superior performance, since CAPA is able to exploit spatial DoFs enabled by the continuous aperture. Furthermore, the proposed KA-based WMMSE method outperforms the state-of-the-art Fourier-based method. These results validate the effectiveness and superiority of the proposed approach for continuous source current design. 


\begin{figure}[!t]
    \centering
    \includegraphics[scale=0.58]{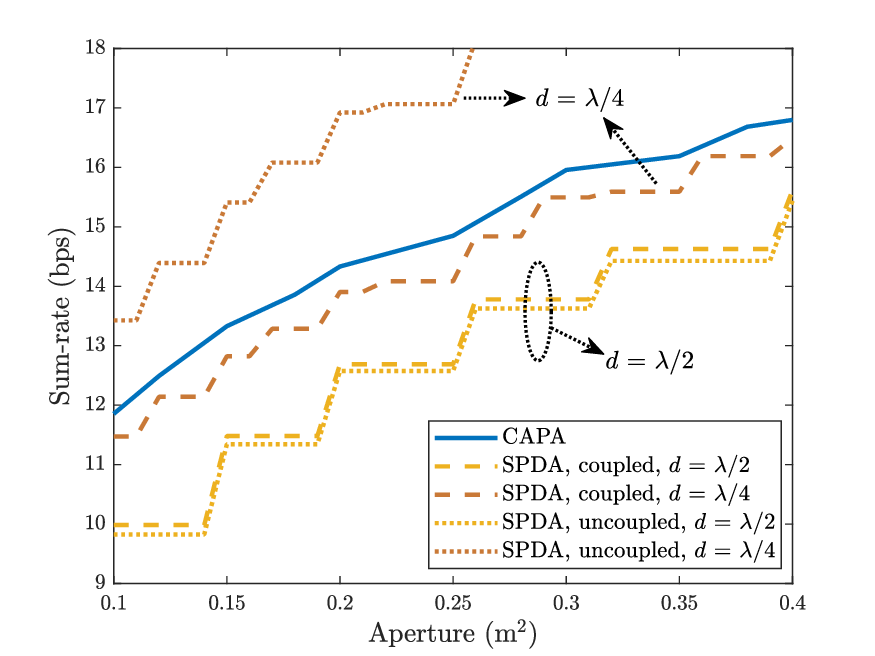}
    \caption{Sum-rate (bps) versus the aperture size (m$^2$).}
    \label{fig:aperture}
\end{figure}

Fig. \ref{fig:aperture} demonstrates the impact of aperture size $A_\mathrm{T}$ on the sum-rate. It can be observed that the sum-rate of all the cases consistently increases as the aperture size grows, since a larger aperture provides more DoFs. Aligned with previous findings, the CAPA-aided system outperforms the  conventional SPDA scenario with mutual coupling effects, and the SPDA case with smaller antenna spacing gradually approaches the CAPA performance in terms of sum-rate. This implies that CAPA is viewed as the extreme case of the dense deployment of array. We also compare scenarios with and without considering the mutual coupling effects. For half-wavelength antenna spacing, the case with and without mutual coupling effects exhibit nearly identical sum-rate performance, indicating that the mutual coupling effect is negligible in this regime. In contrast, for the quarter-wavelength antenna spacing, the case that neglects mutual coupling significantly outperforms both the CAPA-aided system and the SPDA system with mutual coupling, demonstrating that the sum-rate is severely overestimated if the  mutual coupling effects are ignored. 

\begin{figure}[!t]
    \centering
    \includegraphics[scale=0.58]{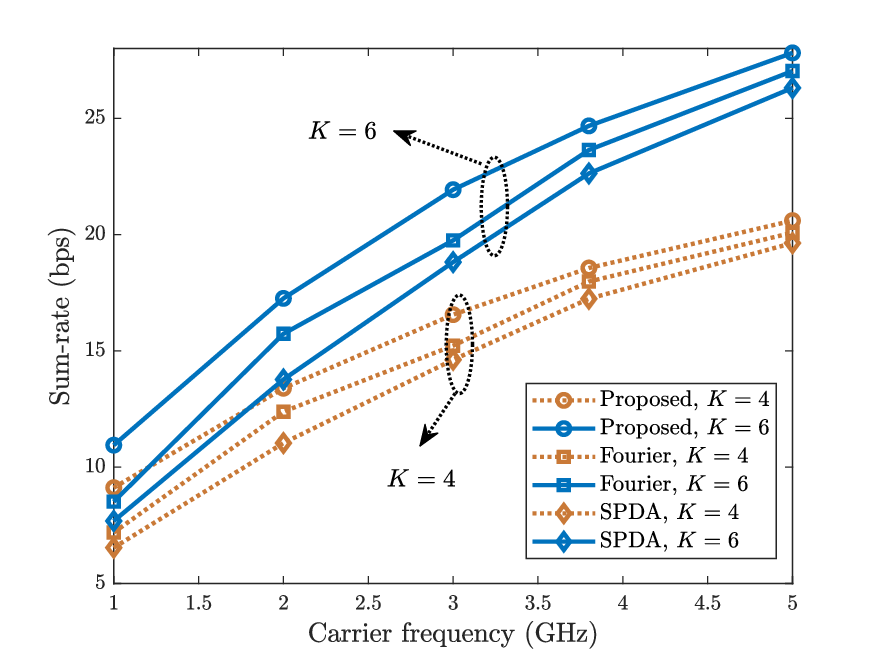}
    \caption{Sum-rate (bps) versus carrier frequency (GHz).}
    \label{fig:carrier}
\end{figure}

Fig. \ref{fig:carrier} illustrates the trend of sum-rate under different user numbers as the carrier frequency $f_c$ increases. As expected, the sum-rate improves with increasing $f_c$, owing to the higher spatial DoFs achieved by the higher frequencies. Moreover, the proposed KA-based WMMSE method for the CAPA-aided system consistently achieves superior performance compared with the benchmark schemes, which is in alignment with the previous results. On the other hand, it can be observed that increasing $K$ results in higher sum-rate performance across all the cases, and the corresponding performance gain becomes more pronounced as the carrier frequency increases. As compared to the case of $K=4$, the sum-rate improves by 2 bps at $f_c=1$ GHz, while the improvement increases to about 7~bps at $f_c = 5$~GHz.  

\begin{figure}[!t]
    \centering
    \includegraphics[scale=0.58]{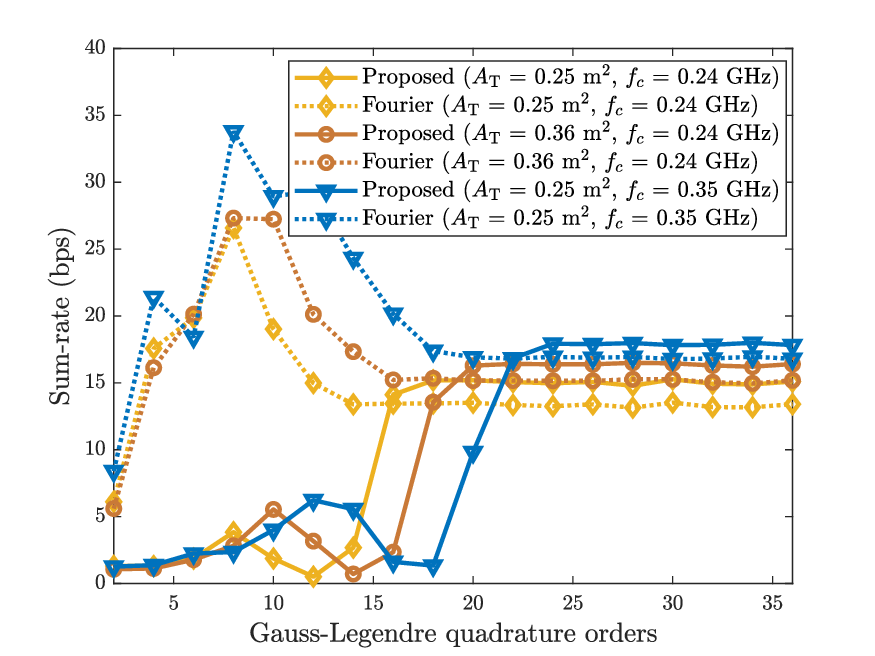}
    \caption{Convergence of the sum-rate with different Gauss-Legendre quadrature orders under different carrier frequencies and aperture sizes.}
    \label{fig:GL}
\end{figure}

In Fig. \ref{fig:GL}, the convergence behaviors with increasing Gauss-Legendre quadrature orders  are illustrated under different aperture sizes and carrier frequencies. It can be observed that all cases gradually converge to stable values as $M$ increases, indicating that an appropriately chosen $M$ can ensure the accuracy of the integrals. Meanwhile,  the proposed algorithm achieves higher sum-rate than the Fourier-based methods once convergence is reached, further validating the correctness of the results. Moreover, it is interesting to see that a larger quadrature order is required for convergence when the aperture size increases or the carrier frequency becomes higher. This behavior arises because larger apertures and higher carrier frequencies result in larger integration domains and more rapid phase oscillations, thereby necessitating a higher quadrature order to accurately approximate the integrals. 

\subsection{CAPA-to-CAPA MIMO Scenario}
In the CAPA-to-CAPA MIMO scenario, the default parameters are specified as follows.  Both the CAPA transmitter and the CAPA receiver are square surfaces with areas of $A_\mathrm{T}=A_\mathrm{R}\!=\!0.25$ m$^2$. The center of the CAPA transmitter and the CAPA receiver are located  at ($0$, $0$, $0$) m and ($10$, $-10$, $0$) m, respectively. The CAPA receiver is parallel to the CAPA transmitter and has no rotation relative to the CAPA transmitter. The transmitter has $N=4$ independent data streams. Unless otherwise specified, all remaining parameters are set the same as those in the multi-user scenario. 

\begin{figure}[!t]
    \centering
    \includegraphics[scale=0.58]{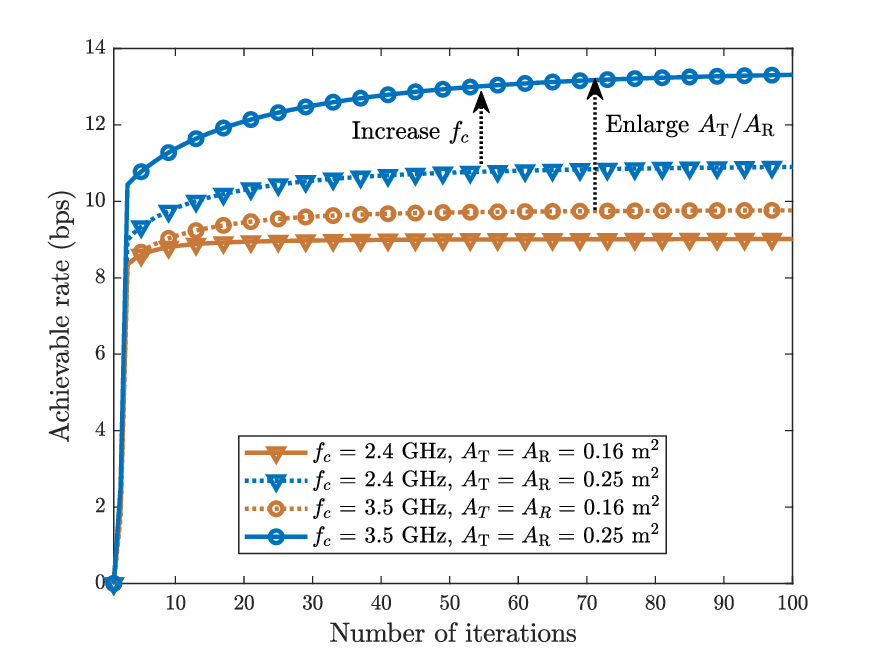}
    \caption{Convergence behavior of the CAPA-to-CAPA MIMO system under different carrier frequencies and aperture sizes.}
    \label{fig:MIMO_Conv}
\end{figure}

In Fig. \ref{fig:MIMO_Conv}, we exhibit the convergence behaviors of the considered CAPA-to-CAPA MIMO system under different aperture sizes and carrier frequencies. It can be observed that the achievable rate converges across all the cases, showing the effectiveness of the extension framework. Besides, we can find that increasing the carrier frequency and the aperture size leads to higher achievable rates after convergence, whereas scenarios with higher carrier frequencies and larger apertures require more iterations to converge. This behavior arises from two main factors: (1) higher spatial DoFs are enabled by larger apertures and higher carrier frequencies; (2) larger aperture size corresponds to larger integral field, while higher carrier frequency leads to faster phase oscillations, both of which slow down the convergence process. 

\begin{figure}[!t]
    \centering
    \includegraphics[scale=0.58]{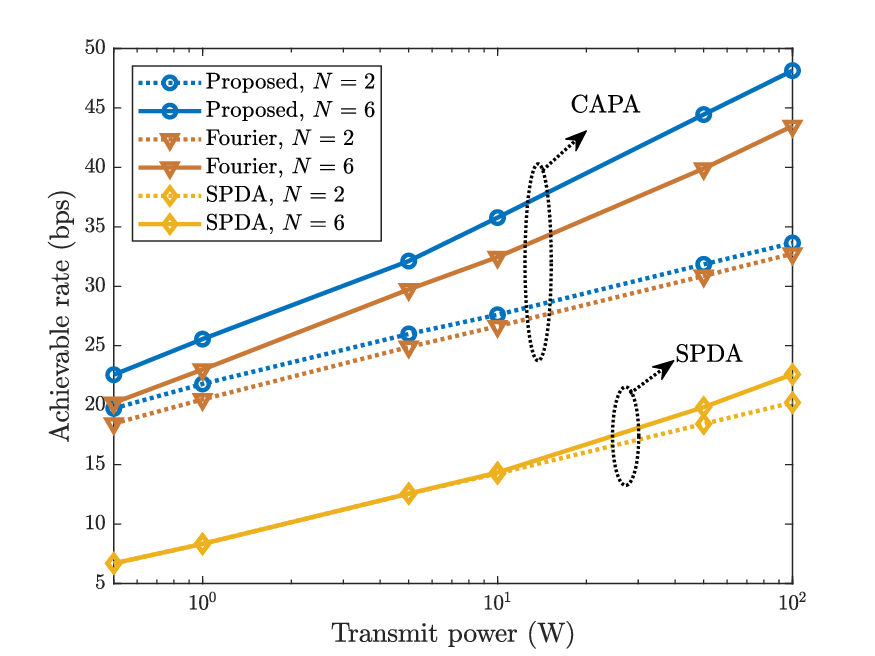}
    \caption{Achievable rate (bps) versus transmit power (W).}
    \label{fig:MIMO_Pt}
\end{figure}

Fig. \ref{fig:MIMO_Pt} shows the effect of the transmit power on the achievable rate of the CAPA-to-CAPA MIMO system, where the CAPA receiver is located at ($3$, $-3$, $0$) m. As shown in Fig. \ref{fig:MIMO_Pt}, the achievable rate increases monotonically with the transmit power, and the proposed method consistently outperforms the Fourier-based method as well as the SPDA schemes across different transmit power. Moreover, it can observed that the case with a larger number of data streams $N=6$ can achieve higher rate compared to the case with $N=2$, and the corresponding performance gap becomes more pronounced when given more transmit power. This behavior can be attributed to the increase in signal-to-noise ratio (SNR) with higher transmit power.  Under higher SNR conditions, more data streams allows the system to better exploit the available DoFs, thereby enlarging the performance gain. 

\begin{figure}[!t]
    \centering
    \includegraphics[scale=0.58]{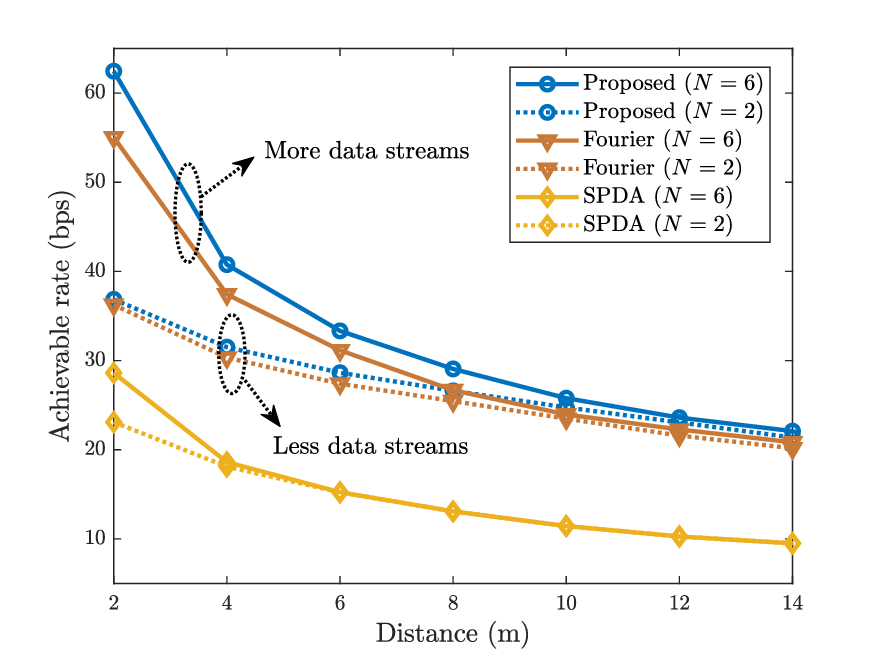}
    \caption{Achievable rate (bps) versus distance (m) between CAPA transmitter and CAPA receiver.}
    \label{fig:MIMO_Dist}
\end{figure}

Fig. \ref{fig:MIMO_Dist} illustrates the trend of achievable rate as the distance between the transmitter and receiver increases. It can be observed that the achievable rate decreases with increasing distance. This behavior is mainly attributed to two factors. On the one hand,  a larger separation distance leads to more severe free-space path loss, resulting in reduced channel gain.  On the other hand, as the distance gradually increases, the near-field effects become weaker and the channels tends to approach a line-of-sight regime, which reduces the available spatial DoFs \cite{Yuanwei_Nearfield}.  Additionally, it is observed that the achievable rates of CAPA with $N=6$ substantially outperformed those with $N=2$ when the transceivers are placed close to each other, which is also attributed to the enhanced spatial DoFs at short distances. 

\section{Conclusions} \label{conclusions}
This paper studied a CAPA-aided multi-user communication system considering the mutual coupling effects, in which a mutual coupling-aware sum-rate maximization functional optimization was formulated. To tackle the resultant problem, a KA-based WMMSE algorithm was developed, where the closed-form solution of optimal beamformer was derived via CoV and KA. Additional, an extension to CAPA-to-CAPA scenario was presented. Numerical results demonstrated that beamforming design considering mutual coupling effects is more physically rational in practical scenarios, and the proposed algorithm is effective and superior as compared to the benchmarks.  

\begin{appendices}   
\section{Proof of Proposition \ref{optimal_structure_prop}}  \label{app_optimal_structure_prop}
Since $f_c$ in (\ref{MMSE_problem_2_3}{b}) is a convex quadratic functional w.r.t. ${\mathbf{w}}(\mathbf{s})$, CoV is can be employed to find the optimal solution. To this end, we first define the variation of $f_c$ as $\Phi(\epsilon)=f_c\left(\delta{\mathbf{w}}(\mathbf{s}) \right)$, where  $\delta{\mathbf{w}}(\mathbf{s}) ={\mathbf{w}}(\mathbf{s}) +\epsilon\boldsymbol{\eta}(\mathbf{s})$ and $\epsilon\boldsymbol{\eta}(\mathbf{s})$ is any perturbation function. Explicitly, $\Phi(\epsilon)$ can be written as
\begin{align}
\Phi(\epsilon)=
    &\int_{\mathcal{S}_{\mathrm{T}}} \int_{\mathcal{S}_{\mathrm{T}}}    \mathbf{g}(\mathbf{z}) \mathbf{U} \mathbf{g}^\mathrm{H}(\mathbf{s})\delta{\mathbf{w}}(\mathbf{s})  \delta{\mathbf{w}}^\mathrm{H}(\mathbf{z})   d\mathbf{s} d\mathbf{z} \notag \\
    &-  2\Re\left\{\int_{\mathcal{S}_{\mathrm{T}}}  \delta{\mathbf{w}}(\mathbf{s}) \mathbf{U} \mathbf{g}^\mathrm{H}(\mathbf{s})  d\mathbf{s}\right\} \notag\\
    &+ \frac{1}{\beta}\int_{\mathcal{S}_{\mathrm{T}}} \int_{\mathcal{S}_{\mathrm{T}}}   \delta{\mathbf{w}}(\mathbf{s}) c_{\mathrm{T}} (\mathbf{s}-\mathbf{z})  \delta{\mathbf{w}}^\mathrm{H}(\mathbf{s})   d \mathbf{z} d \mathbf{s} \notag \\
    = & ~2\epsilon \Re\left\{\int_{\mathcal{S}_{\mathrm{T}}}   \boldsymbol{\eta}(\mathbf{s}){\mathbf{p}}^\mathrm{H}(\mathbf{s}) d \mathbf{s}\right\} + \Psi(\epsilon^2) + \widetilde{C},
\end{align}
where $\Psi(\epsilon^2)$ collects all the terms with $\epsilon^2$, while $\widetilde{C}$ denotes the constant terms that are independent of $\epsilon$. In addition, we denote 
\begin{align}
      {\mathbf{p}} (\mathbf{s})=&\int_{\mathcal{S}_{\mathrm{T}}}  \mathbf{g}(\mathbf{s})   \mathbf{U} \mathbf{g}^\mathrm{H}(\mathbf{z})  {\mathbf{w}} (\mathbf{z}) d\mathbf{z} - \mathbf{g} (\mathbf{s})  \mathbf{U} \notag  \\
      &+ \frac{1}{\beta}  \int_{\mathcal{S}_{\mathrm{T}}}   c_{\mathrm{T}} (\mathbf{s}-\mathbf{z}){\mathbf{w}} (\mathbf{z}) d\mathbf{z}.
\end{align}
Accordingly, the optimal ${\mathbf{w}}(\mathbf{s})$ can be achieved by enforcing the following optimality condition 
\begin{align} \label{CoV_condition}
    \left.\frac{d}{d \epsilon} \Phi(\epsilon)\right|_{\epsilon=0} =2  \Re\left\{\int_{\mathcal{S}_{\mathrm{T}}}   \boldsymbol{\eta}(\mathbf{s}){\mathbf{p}}^\mathrm{H}(\mathbf{s}) d \mathbf{s}\right\}=0.
\end{align}
Based on the fundamental lemma of CoV \cite{Zhaolin_CAPA_BF}, for any arbitrary given $\boldsymbol{\eta}(\mathbf{s})$, the condition ${\mathbf{p}}(\mathbf{s})=0$ must hold to satisfy the condition (\ref{CoV_condition}), which yields 
    \begin{align}
      \!\int_{\mathcal{S}_{\mathrm{T}}} \! \!\! c_{\mathrm{T}} (\mathbf{s}\!-\!\mathbf{z})  {\mathbf{w}}(\mathbf{z})d\mathbf{z}\!=\!\beta\mathbf{g}(\mathbf{s}) \mathbf{U} \left(\mathbf{I}\!-\!\int_{\mathcal{S}_{\mathrm{T}}} \! \! \!\mathbf{g}^\mathrm{H}(\mathbf{z})   {\mathbf{w}}  (\mathbf{z})d\mathbf{z} \! \right).
\end{align}

This completes the derivation.

\section{Proof of Proposition \ref{closed_form_prop}}  \label{app_closed_form_prop}
Given the $c_{\mathrm{T}}^{-1} (\mathbf{s}'-\mathbf{s})$, the closed-form solution of optimal ${\mathbf{w}}(\mathbf{s}')$ can be derived as follows. In (\ref{optimal_structure}), we first multiply both sides by $c_{\mathrm{T}}^{-1} (\mathbf{s}'-\mathbf{s})$ and then integrate over $\mathbf{s}$. This yields  
\begin{align}\label{optimal_structure_2}
     {\mathbf{w}}(\mathbf{s}') =  \beta\widetilde{\mathbf{g}}(\mathbf{s}') \mathbf{U} \left(\mathbf{I}-\int_{\mathcal{S}_{\mathrm{T}}}  \mathbf{g}^\mathrm{H}(\mathbf{z})   {\mathbf{w}}  (\mathbf{z})d\mathbf{z}  \right),
\end{align}
where we define 
\begin{align}
    \widetilde{\mathbf{g}}(\mathbf{s}')=\int_{\mathcal{S}_{\mathrm{T}}} c_{\mathrm{T}}^{-1} (\mathbf{s}'-\mathbf{s}) \mathbf{g}(\mathbf{s})d\mathbf{s}.
\end{align}
However, it can be observed from (\ref{optimal_structure_2}) that  ${\mathbf{w}}(\mathbf{z})$ is still appears inside $\int_{\mathcal{S}_{\mathrm{T}}}  \mathbf{g}^\mathrm{H}(\mathbf{z})   {\mathbf{w}}  (\mathbf{z})d\mathbf{z}$, which prevents the isolation of ${\mathbf{w}}(\mathbf{s})$ from other components. To address it,  we further multiply both sides by $\mathbf{g}^\mathrm{H}(\mathbf{s}')$ and integrate
over $\mathbf{s}'$, yielding 
\begin{align} \label{eq_long}
    & \int_{\mathcal{S}_{\mathrm{T}}}\mathbf{g}^\mathrm{H}(\mathbf{s}') {\mathbf{w}}(\mathbf{s}')d\mathbf{s}' \notag \\
    &=  \!\beta\int_{\mathcal{S}_{\mathrm{T}}}\!\mathbf{g}^\mathrm{H}(\mathbf{s}')\widetilde{\mathbf{g}}(\mathbf{s}') d\mathbf{s}'  \mathbf{U} \left(\mathbf{I}\!-\!\int_{\mathcal{S}_{\mathrm{T}}} \! \mathbf{g}^\mathrm{H}(\mathbf{z})   {\mathbf{w}}  (\mathbf{z})d\mathbf{z}  \right).
\end{align}
By defining
\begin{subequations} \label{Lambda_G}
 \begin{align}
    \boldsymbol{\Lambda} &= \int_{\mathcal{S}_{\mathrm{T}}}\mathbf{g}^\mathrm{H}(\mathbf{s}') {\mathbf{w}}(\mathbf{s}')d\mathbf{s}',\\
    \widetilde{\mathbf{G}}&=\int_{\mathcal{S}_{\mathrm{T}}}\mathbf{g}^\mathrm{H}(\mathbf{s}')\widetilde{\mathbf{g}}(\mathbf{s}') d\mathbf{s}' \notag\\
    &=\int_{\mathcal{S}_{\mathrm{T}}}\int_{\mathcal{S}_{\mathrm{T}}}\mathbf{g}^\mathrm{H}(\mathbf{s}')c_{\mathrm{T}}^{-1} (\mathbf{s}'-\mathbf{s}) \mathbf{g}(\mathbf{s}) d\mathbf{s}'d\mathbf{s},
\end{align}   
\end{subequations}
the equation (\ref{eq_long}) can be re-written as 
\begin{align} \label{matrix_isolated}
    \boldsymbol{\Lambda}&=\beta\widetilde{\mathbf{G}}  \mathbf{U} \left(\mathbf{I}-\boldsymbol{\Lambda} \right) \notag\\
    \Longleftrightarrow \boldsymbol{\Lambda}&=\left(\frac{1}{\beta}\mathbf{I}+\widetilde{\mathbf{G}} \mathbf{U}\right)^{-1}\widetilde{\mathbf{G}}  \mathbf{U}.
\end{align}
By substituting (\ref{Lambda_G}{a}) and (\ref{matrix_isolated}) to (\ref{optimal_structure_2}), the closed-form solution of ${\mathbf{w}}(\mathbf{s}')$ can be obtained as follows:
\begin{align}
     {\mathbf{w}}(\mathbf{s}') &=  \beta\widetilde{\mathbf{g}}(\mathbf{s}') \mathbf{U} \left(\mathbf{I}-\left(\frac{1}{\beta}\mathbf{I}+\widetilde{\mathbf{G}} \mathbf{U}\right)^{-1}\widetilde{\mathbf{G}}  \mathbf{U}  \right) \notag \\
    & \overset{\mathrm{(a)}}{=} \widetilde{\mathbf{g}}(\mathbf{s}') \mathbf{U} \left(\frac{1}{\beta}\mathbf{I}+\widetilde{\mathbf{G}} \mathbf{U}\right)^{-1}.
\end{align}
In (a), the Woodbury Identity has been employed. The proof ends.

\section{Proof of Proposition \ref{FT_prop}}  \label{app_FT_prop}
According to the variation of the Euler's formula $\sin(\theta) =  \left(e^{\jmath\theta}-e^{-\jmath\theta}\right)/2\jmath$, $\phi(\mathbf{s})$ can be rewritten as
\begin{align} \label{phi_Euler}
    \phi(\mathbf{s}) =\frac{\sin(\kappa_0 \|\mathbf{s}\|)}{4\pi \|\mathbf{s}\|}=\frac{1}{8\pi \jmath} \cdot \left(\frac{e^{\jmath\kappa_0 \|\mathbf{s}\|}}{\|\mathbf{s}\|}-\frac{e^{-\jmath\kappa_0 \|\mathbf{s}\|}}{\|\mathbf{s}\|}\right).
\end{align}
To find $\mathcal{F}\{\phi(\mathbf{s})\}$, the Fourier transform of ${e^{\jmath\kappa_0 \|\mathbf{s}\|}}/{\|\mathbf{s}\|}$ is required. According to the Weyl identity, we have
\begin{align}
\frac{e^{\jmath \kappa_0 \|\mathbf{s}\|}}{\|\mathbf{s}\|} = \frac{\jmath}{2\pi} \iint_{-\infty}^{\infty}  e^{\jmath \boldsymbol{\kappa}^\mathrm{T} \mathbf{s}} \frac{e^{ \jmath  \kappa_z |s_z|}}{\kappa_z} d\kappa_x d\kappa_y,
\end{align}
where $\mathbf{s}=[s_x, s_y, s_z]^\mathrm{T}$, $\kappa_0=\|[\kappa_x,\kappa_y,\kappa_z]^\mathrm{T}\|=2\pi/\lambda$ and $\kappa_z=\sqrt{\kappa_0^2-\|\boldsymbol{\kappa}\|^2}\geq0$. Considering the inverse transform in (\ref{FT}{b}), the Fourier transform of ${e^{\jmath\kappa_0 \|\mathbf{s}\|}}/{\|\mathbf{s}\|}$ can be given by 
\begin{align} \label{FT_instance_2}
    \mathcal{F}\left\{\frac{e^{\jmath \kappa_0 \|\mathbf{s}\|}}{\|\mathbf{s}\|}\right\} &= {\jmath}{2\pi}\frac{e^{ \jmath  \kappa_z |s_z|}}{\sqrt{\kappa_0^2-\|\boldsymbol{\kappa}\|^2}} \overset{(\text{a})}{=}\frac{\jmath 2\pi}{\sqrt{\kappa_0^2-\|\boldsymbol{\kappa}\|^2}}.
\end{align}
Since the CAPA is placed in $x-y$ plane, $s_z=0$, where (a) holds. According to the conjugation property of Fourier transform, the Fourier transform of ${e^{-\jmath\kappa_0 \|\mathbf{s}\|}}/{\|\mathbf{s}\|}$ is
\begin{align}\label{FT_instance_1}
    \mathcal{F}\left\{\frac{e^{-\jmath \kappa_0 \|\mathbf{s}\|}}{\|\mathbf{s}\|}\right\} &= \frac{-\jmath 2\pi}{\sqrt{\kappa_0^2-\|\boldsymbol{\kappa}\|^2}}.
\end{align}
Substitute (\ref{FT_instance_1}) and (\ref{FT_instance_2}) to (\ref{phi_Euler}), we can obtain 
\begin{align}
    \mathcal{F}\{\phi(\mathbf{s})\} = \begin{cases}
        \frac{1}{2\sqrt{\kappa_0^2-\|\boldsymbol{\kappa}\|^2}},& \|\boldsymbol{\kappa}\| \leq \kappa_0, \\
        0,& \|\boldsymbol{\kappa}\| \geq \kappa_0.
    \end{cases}
    \end{align}

The derivation has been completed.
\end{appendices}

\bibliographystyle{IEEEtran}
\bibliography{references2}{}

\end{document}